\newif\ifCONF{}
\newif\ifFULL{}

\CONFfalse
\FULLtrue

\ifCONF{}
    \documentclass[compsoc,conference,a4paper,10pt,times]{IEEEtran}
    \IEEEoverridecommandlockouts
    \def\BibTeX{{\rm B\kern-.05em{\sc i\kern-.025em b}\kern-.08em
        T\kern-.1667em\lower.7ex\hbox{E}\kern-.125emX}}
\else
    \documentclass[a4paper,english,11pt]{article}
    \usepackage{fullpage}

    \usepackage[toc]{multitoc}

\fi

\ifFULL{}
    \usepackage[numbers]{natbib}
    \usepackage{authblk}
\fi

\usepackage{booktabs}
\usepackage{amsmath}
\usepackage{amsthm}
\usepackage{algorithm}
\usepackage{microtype}
\usepackage{url}
\usepackage[inline]{enumitem}
\usepackage[colorlinks=true,urlcolor=black]{hyperref}
\usepackage{cleveref}  % Load me after `hyperref`
\usepackage{graphicx}
\usepackage{xifthen}

\usepackage[
    lambda,
    advantage,
    operators,
    sets,
    adversary,
    landau,
    probability,
    notions,
    logic,
    ff,
    mm,
    primitives,
    events,
    complexity,
    asymptotics,
    keys
]{cryptocode}

\theoremstyle{plain}
\newtheorem{theorem}{Theorem}[section]
\newtheorem{lemma}{Lemma}[section]
\newtheorem*{remark}{Remark}

\theoremstyle{definition}
\newtheorem{definition}{Definition}[section]
\newtheorem{property}{Property}[section]

\crefname{theorem}{theorem}{theorems}
\Crefname{theorem}{Theorem}{Theorems}
\crefname{lemma}{lemma}{lemmas}
\Crefname{lemma}{Lemma}{Lemmas}
\crefname{remark}{remark}{remarks}
\Crefname{remark}{Remark}{Remarks}
\crefname{definition}{definition}{definitions}
\Crefname{definition}{Definition}{Definitions}
\crefname{property}{property}{properties}
\Crefname{property}{Property}{Properties}

\newlist{inlinelist}{enumerate*}{1}
\setlist[inlinelist]{label=\arabic*)}

\newcommand{\rv}[1]{\mathbf{#1}}

\newcommand{\Game}{\ensuremath{\rv{G}}}

\newcommand{\Set}{\pcalgostyle{Set}}
\newcommand{\Derive}{\pcalgostyle{Derive}}
\newcommand{\DerivePriv}{\pcalgostyle{DPriv}}
\newcommand{\DerivePub}{\pcalgostyle{DPub}}
\newcommand{\Sign}{\pcalgostyle{Sign}}

\newcommand{\Verify}{\pcalgostyle{Vrfy}}

\newcommand{\SetDerivePPSignVerify}{\ensuremath{(\Set, \DerivePub, \DerivePriv, \allowbreak \Sign, \Verify)}}
\newcommand{\Pub}{\pcalgostyle{Pub}}
\newcommand{\Ver}{\pcalgostyle{Ver}}

\newcommand{\Sec}{\pcalgostyle{Sec}}
\newcommand{\PubSec}{\ensuremath{(\Pub, \Sec)}}

\newcommand{\ppssppkk}{\ensuremath{(\pp{}, \{\derkeyi\}_{\cvi\in\cV})}}

\newcommand{\Keygen}{\pcalgostyle{Keygen}}
\newcommand{\Hyb}{\ensuremath{\Game{}}}

\newcommand{\crtset}{\ensuremath{\set{\hatsigmai}_{\cvi \in \cV}}}

\newcommand{\pkK}{\ensuremath{\wpk}}

\newcommand{\pkX}{\ensuremath{\overline{\pk}}}
\newcommand{\pkXi}{\ensuremath{\pkX_{i}}}
\newcommand{\skX}{\ensuremath{\overline{\sk}}}
\newcommand{\skXi}{\ensuremath{\skX_{i}}}

\newcommand{\HybZero}{\ensuremath{\Hyb_0}}

\newcommand{\PRFOracle}{\ensuremath{\Oracle_{\fash}}}

\newcommand{\Rand}{\ensuremath{\mathcal{F}_{\secpar}}}

\newcommand{\CorrOracle}{\ensuremath{\Oracle_{\corr}}}

\newcommand{\SignOracle}{\ensuremath{\Oracle_{\Sign}}}

\newcommand{\GamePRFZero}{\ensuremath{\Game{}_{\fash{},\Adv}^{\mathsf{prf}-0}(\secpar)}}
\newcommand{\GamePRFOne}{\ensuremath{\Game{}_{\fash{}, \Adv}^{\mathsf{prf}-1}(\secpar)}}
\newcommand{\GameHSigForge}{\ensuremath{\Game{}_{\Pi, \Adv}^{\mathsf{heuf}}(\secpar, \cG{}})}
\newcommand{\GameSigForge}{\ensuremath{\Game{}_{\Pi, \Adv}^{\mathsf{euf}}(\secpar})}
\newcommand{\GameSigForgeSig}{\ensuremath{\Game{}_{\Sig, \Adv'}^{\mathsf{euf}}(\secpar})}

\newcommand{\GameSKeyInd}{\ensuremath{\Game{}_{\Pi, \Adv}^{\mathsf{sk-ind}}(\secpar, \cG{}})}
\newcommand{\GameSKeyIndDHKA}{\ensuremath{\Game{}_{\DHKA, \Adv}^{\mathsf{sk-ind}}(\secpar, \cG{}})}
\newcommand{\GameSKeyIndB}{\ensuremath{\Game{}_{\DHKA{}, \Adv}^{\mathsf{sk-ind}-b}(\secpar, \cG{}})}
\newcommand{\GameSKeyIndZero}{\ensuremath{\Game{}_{\DHKA{}, \Adv}^{\mathsf{sk-ind}-0}(\secpar, \cG{}})}
\newcommand{\GameSKeyIndOne}{\ensuremath{\Game{}_{\DHKA{}, \Adv}^{\mathsf{sk-ind}-1}(\secpar, \cG{}})}

\newcommand{\GameEncPi}{\ensuremath{\Game{}_{\Pi, \Adv}^{\mathsf{sem}}(\secpar})}

\newcommand{\Sig}{\ensuremath{\Sigma}}

\newcommand{\DHKA}{\ensuremath{\Gamma}}
\newcommand{\sDHKA}{\DHKA{} = \SetDeriveDHKA{}}

\newcommand{\SetDHKA}{\ensuremath{\Set_{\DHKA}}}
\newcommand{\DeriveDHKA}{\ensuremath{\Derive_{\DHKA}}}

\newcommand{\SetDeriveDHKA}{\ensuremath{(\SetDHKA, \DeriveDHKA)}}

\newcommand{\abort}{\ensuremath{\mathsf{abort}}}

\newcommand{\eabort}{\ensuremath{E_{\abort}}}

\newcommand{\enabort}{\ensuremath{\lnot E_{\abort}}}

\newcommand{\Adv}{\pcalgostyle{A}}
\newcommand{\Advfirst}{\ensuremath{\pcalgostyle{A}'}}

\newcommand{\Dist}{{\ensuremath{\pcalgostyle{D}}}}
\newcommand{\Oracle}{\ensuremath{\pcalgostyle{O}}}

\newcommand{\cl}[1]{%
    \ifthenelse{\isempty{#1}}
    {\ensuremath{l}}
    {\ensuremath{l_{#1}}}
}
\newcommand{\cli}{\ensuremath{\cl{i}}}

\newcommand{\clj}{\ensuremath{\cl{j}}}

\newcommand{\clhfirst}{\ensuremath{\cl{h}'}}
\newcommand{\dlhfirst}{{\ensuremath{\clhfirst{} = \cvh \concat \cwh'}}}
\newcommand{\clzero}{\ensuremath{\cl{0}}}
\newcommand{\dlzero}{\ensuremath{\cv{0}}}
\newcommand{\sclzero}{\ensuremath{\clzero = \dlzero}}
\newcommand{\dli}{\ensuremath{\cvi}}
\newcommand{\dlvi}{\ensuremath{\cvi \concat \cwi}}
\newcommand{\sli}{\ensuremath{\cli = \dli}}
\newcommand{\sliver}{\ensuremath{\cli = \dli \concat \cw{i}}}

\newcommand{\cw}[1]{%
    \ifthenelse{\isempty{#1}}
    {\ensuremath{w}}
    {\ensuremath{w_{#1}}}
}
\newcommand{\cwi}{\ensuremath{\cw{i}}}
\newcommand{\cwj}{\ensuremath{\cw{j}}}
\newcommand{\cwh}{\ensuremath{\cw{h}}}

\newcommand{\ct}[1]{%
    \ifthenelse{\isempty{#1}}
    {\ensuremath{t}}
    {\ensuremath{t_{#1}}}
}
\newcommand{\cti}{\ensuremath{\ct{i}}}
\newcommand{\ctj}{\ensuremath{\ct{j}}}
\newcommand{\cthfirst}{\ensuremath{\ct{h}'}}
\newcommand{\dt}[2]{\ensuremath{\fash_{#1}(00 \concat #2)}}
\newcommand{\dthfirst}{\ensuremath{\dt{\cSh'}{\clhfirst}}}
\newcommand{\dti}{\ensuremath{\dt{\cSi}{\cli}}}
\newcommand{\sti}{\ensuremath{\cti = \dti}}
\newcommand{\sthfirst}{\ensuremath{\cthfirst = \dthfirst}}

\newcommand{\crr}[1]{%
    \ifthenelse{\isempty{#1}}
    {\ensuremath{r}}
    {\ensuremath{r_{#1}}}
}
\newcommand{\crij}{\crr{ij}}

\newcommand{\dr}[2]{\ensuremath{\fash_{#1}(10 \concat #2)}}

\newcommand{\srij}{\ensuremath{\crij = \dr{\cti}{\clj}}}

\newcommand{\cv}[1]{%
    \ifthenelse{\isempty{#1}}
    {\ensuremath{v}}
    {\ensuremath{v_{#1}}}
}
\newcommand{\cvi}{\ensuremath{\cv{i}}}
\newcommand{\cvt}{\ensuremath{\cv{t}}}
\newcommand{\cvtone}{\ensuremath{\cv{t-1}}}

\newcommand{\cvj}{\ensuremath{\cv{j}}}
\newcommand{\cvp}{\ensuremath{\cv{p}}}
\newcommand{\cvpfirst}{\ensuremath{\cv{p'}}}
\newcommand{\cvh}{\ensuremath{\cv{h}}}

\newcommand{\cvstar}{\ensuremath{\cv{}^{*}}}

\newcommand{\cvzero}{\ensuremath{\cv{0}}}
\newcommand{\cvroot}{\ensuremath{\cv{R}}}

\newcommand{\cTTi}{\ensuremath{\cTT_{i}}}
\newcommand{\cTTj}{\ensuremath{\cTT_{j}}}

\newcommand{\cy}[1]{%
    \ifthenelse{\isempty{#1}}
    {\ensuremath{y}}
    {\ensuremath{y_{#1}}}
}
\newcommand{\cyij}{\ensuremath{\cy{ij}}}
\newcommand{\dy}[3]{\ensuremath{\enc_{#1}{(#2 \concat #3)}}}
\newcommand{\dyij}{\ensuremath{\dy{\crij}{\ctj}{\cxj}}}
\newcommand{\syij}{\ensuremath{\cyij \sample \dyij}}
\newcommand{\dykh}{\ensuremath{\dy{\crr{kh}}{\cthfirst}{\cxhfirst}}}
\newcommand{\sykh}{\ensuremath{\cy{kh}' \sample \dykh{}}}

\newcommand{\cS}[1]{%
    \ifthenelse{\isempty{#1}}
    {\ensuremath{S}}
    {\ensuremath{S_{#1}}}
}
\newcommand{\cSi}{\ensuremath{\cS{i}}}

\newcommand{\cSj}{\ensuremath{\cS{j}}}
\newcommand{\cSh}{\ensuremath{\cS{h}}}
\newcommand{\cShfirst}{\ensuremath{\cSh'}}
\newcommand{\cSeed}{\ensuremath{\cS{}}}
\newcommand{\sSeed}{\ensuremath{\cSeed \in \cSS}}
\newcommand{\cSzero}{\ensuremath{\cS{0}}}

\newcommand{\dS}[2]{\ensuremath{\fash_{#1}{(11 \concat #2)}}}
\newcommand{\dSzero}{\ensuremath{\dS{\cSeed}{\clzero}}}

\newcommand{\dSi}{\ensuremath{\dS{\cSj}{\cli}}}
\newcommand{\dShfirst}{\ensuremath{\dS{\cS{p}}{\clhfirst}}}
\newcommand{\scSzero}{\ensuremath{\cSzero = \dSzero{}}}
\newcommand{\scSi}{\ensuremath{\cSi = \dSi}}
\newcommand{\sShfirst}{\ensuremath{\cShfirst = \dShfirst}}

\newcommand{\ck}[1]{%
    \ifthenelse{\isempty{#1}}
    {\ensuremath{k}}
    {\ensuremath{k_{#1}}}
}
\newcommand{\cki}{\ensuremath{\ck{i}}}

\newcommand{\ski}{\ensuremath{\cki = \fash_{\cSi}(01 \concat \cli)}}

\newcommand{\msk}{\ensuremath{\mathsf{msk}}}
\newcommand{\mski}{\ensuremath{\msk_{i}}}

\newcommand{\mpk}{\ensuremath{\mathsf{mpk}}}
\newcommand{\mpki}{\ensuremath{\mpk_{i}}}
\newcommand{\wpk}{\ensuremath{\mathsf{wpk}}}

\newcommand{\sski}{\ensuremath{\sk_{i}}}
\newcommand{\sskj}{\ensuremath{\sk_{j}}}

\newcommand{\ppki}{\ensuremath{\pk_{i}}}
\newcommand{\ppkj}{\ensuremath{\pk_{j}}}
\newcommand{\mspk}{\ensuremath{(\msk, \mpk)}}

\newcommand{\ssppkifirst}{\ensuremath{(\skXi, \pkXi)}}

\newcommand{\cx}[1]{%
    \ifthenelse{\isempty{#1}}
    {\ensuremath{x}}
    {\ensuremath{x_{#1}}}
}
\newcommand{\cxi}{\ensuremath{\cx{i}}}
\newcommand{\cxj}{\ensuremath{\cx{j}}}

\newcommand{\cxhfirst}{\ensuremath{\cx{h}'}}
\newcommand{\cxstar}{\ensuremath{\cx{\cvstar}}}
\newcommand{\cxstarbar}{\ensuremath{\bar{x}_{\cvstar}}}

\newcommand{\dx}[2]{\ensuremath{\fash_{#1}(01 \concat #2)}}
\newcommand{\dxi}{\ensuremath{\dx{\cSi}{\cli}}}
\newcommand{\dxhfirst}{\ensuremath{\dx{\cSh'}{\clhfirst}}}
\newcommand{\sxi}{\ensuremath{\cxi = \dxi}}
\newcommand{\sxj}{\ensuremath{\cxj = \dxi}}
\newcommand{\sxhfirst}{\ensuremath{\cxhfirst = \dxhfirst}}

\newcommand{\cK}[1]{%
    \ifthenelse{\isempty{#1}}
    {\ensuremath{K}}
    {\ensuremath{K_{#1}}}
}

\newcommand{\cX}[1]{%
    \ifthenelse{\isempty{#1}}
    {\ensuremath{X}}
    {\ensuremath{X_{#1}}}
}

\newcommand{\csxi}{\ensuremath{(\cSi, \cxi)}}

\newcommand{\cG}{\ensuremath{G}}
\newcommand{\cGfirst}{\ensuremath{G'}}
\newcommand{\cGt}{\ensuremath{\cG{}_{\cTT{}}}}

\newcommand{\cE}{\ensuremath{E}}

\newcommand{\cV}{\ensuremath{V}}

\newcommand{\cSS}{\ensuremath{\mathcal{S}}}
\newcommand{\sG}{\ensuremath{\cG = (\cV, \cE)}}
\newcommand{\sE}{\ensuremath{(\cvi, \cvj) \in \cE}}

\newcommand{\cQQ}{\ensuremath{\mathcal{Q}}}

\newcommand{\cXX}{\ensuremath{\mathcal{X}}}
\newcommand{\cMM}{\ensuremath{\mathcal{M}}}

\newcommand{\cTT}{\ensuremath{\mathcal{T}}}
\newcommand{\cPP}{\ensuremath{\mathcal{P}}}

\newcommand{\sDerive}{\ensuremath{\Derive(\cG{}, \Pub{}, \cvi{}, \cvj{}, \cSi{})}}

\newcommand{\swDeriveDHKAi}{\ensuremath{\Derive_{\DHKA{}}(\cG{}, \Pub{}, \cvi{}, \cvi{}, \cSi{})}}
\newcommand{\swDerivePriv}{\ensuremath{\DerivePriv(\pp{}, \allowbreak\derkeyi{}, \cvi{}, \cvj{})}}

\newcommand{\swDerivePub}{\ensuremath{\DerivePub(\pp{}, \cvi{})}}

\newcommand{\sdSet}{\ensuremath{\Set(\secparam, \cG,\allowbreak \cSeed)}}
\newcommand{\sdSetDHKA}{\ensuremath{\Set_{\DHKA}(\secparam, \cG, \cSeed)}}

\newcommand{\sSign}{\ensuremath{\Sign{(\sski, m)}}}
\newcommand{\sSignP}{\ensuremath{\Sign_{\sski}(m)}}
\newcommand{\sVerify}{\ensuremath{\Verify{(\ppki, m, \sigma)}}}

\newcommand{\sVerifyJ}{\ensuremath{\Verify{(\ppkj, m, \Sign(\sskj, m))}}}

\newcommand{\ssigma}{\ensuremath{\sigma \sample \sSignP}}
\newcommand{\hatsigma}{\ensuremath{\cert}}
\newcommand{\cert}{\ensuremath{\mathsf{cert}}}
\newcommand{\hatsigmai}{\ensuremath{\cert_{i}}}
\renewcommand{\pp}{\ensuremath{\mathsf{pp}}}
\renewcommand{\Keygen}{\ensuremath{\mathsf{KGen}}}

\newcommand{\ssDerivePrivArculaZero}{\ensuremath{\sski = \DerivePriv(\pp{}, \derkey_0, \cvzero{}, \cvi{})}}

\newcommand{\sPub}{\ensuremath{\Pub : \cV \cup \cE \to	\bin^{*}}}

\newcommand{\dPub}{\ensuremath{\Pub : \cvi \mapsto \cli \qquad \Pub : (\cvi, \cvj) \mapsto \cyij}}

\newcommand{\sSec}{\ensuremath{\Sec : \cV \to \bin^\secpar \times \bin^\secpar}}

\newcommand{\dSec}{\ensuremath{\Sec : \cvi \mapsto (\cSi, \cxi)}}

\newcommand{\aSec}{\ensuremath{\Sec(\cvi)}}

\newcommand{\prff}{\ensuremath{\set{\fash_{k}}_{k \in \cKK_\secpar}}}

\newcommand{\Epsilon}{\ensuremath{\mathcal{E}}}
\newcommand{\Gen}{\pcalgostyle{Gen}}
\newcommand{\sGen}{\ensuremath{\Gen{(\secparam)}}}
\newcommand{\sEnc}{\ensuremath{\enc(\sk, m)}}
\newcommand{\sDec}{\ensuremath{\dec(\sk, c)}}

\newcommand{\cGroupg}{\ensuremath{g}}
\newcommand{\cGroupq}{\ensuremath{q}}

\newcommand{\cGroupPub}[1]{\ensuremath{\cGroupg{}^{#1}}}

\newcommand{\ckipub}{\ensuremath{\ppki = \cGroupPub{\sski}}}

\newcommand{\ECDSA}{\ensuremath{\Sigma}}
\newcommand{\KeygenECDSA}{\ensuremath{\Keygen_{\ECDSA}}}

\newcommand{\SignECDSA}{\ensuremath{\Sign_{\ECDSA}}}

\newcommand{\VerifyECDSA}{\ensuremath{\Verify_{\ECDSA}}}

\newcommand{\dECDSA}{\ensuremath{\ECDSA{} = \ensuremath{(\KeygenECDSA, \SignECDSA, \allowbreak \VerifyECDSA)}}}

\newcommand{\anc}{\ensuremath{Anc}}
\newcommand{\ancstar}{\ensuremath{\anc{(\cvstar)}}}

\newcommand{\danc}{\ensuremath{\anc(\cvi)} = \set{\cvj{} \mid \cvj{} \leadsto_w \cvi{}}}
\newcommand{\desc}{\ensuremath{Desc}}
\newcommand{\descvi}{\desc(\cvi)}

\newcommand{\ddesc}{\ensuremath{\desc(\cvi)} = \set{\cvj{} \mid \cvi{} \leadsto_w \cvj{}}}

\newcommand{\tDerive}{
    The deterministic derivation algorithm takes as input the access graph \cG{}, the public information \Pub{}, a source node \cvi{}, a target node \cvj{}, and the secret information \cSi{} of node \cvi{}. It outputs the cryptographic key \cx{j} associated to node \cvj{} if $\cvj \in \desc(\cvi)$.
}

\newcommand{\rekey}{
    \begin{enumerate}
        \item Increase the version \cwh{} of the node \cvh{} to a new value $\cwh{}'$ and update the \Ver{} data structure.
            Then, compute a new label \dlhfirst{} and update the corresponding entry in \Pub{}.
            Finally, compute a new secret \sShfirst{}, a new pair of secret and intermediate keys \sxhfirst{} and \sthfirst{}, and update the \Sec{} mapping.
        \item For each incoming edge $(\cv{k}, \cvh{})$ of \cvh{}, update the public information \sykh{} stored on the edge to reflect the updated values \cthfirst{} and \cxhfirst{}.
    \end{enumerate}
}

\newtheorem{construction}{Construction}

\newcommand{\cKK}{\ensuremath{\mathcal{K}}}
\newcommand{\cY}{\ensuremath{\mathcal{Y}}}

\newcommand{\pkforgery}{\pk^\bullet}

\newcommand{\mforgery}{\widetilde{m}}
\newcommand{\sigmaforgery}{\sigma^\bullet}
\newcommand{\sigmaforgerytilde}{\widetilde{\sigma}}
\newcommand{\hatsigmaforgery}{\hatsigma^\bullet}

\newcommand{\OPDUP}{\texttt{OP\_DUP}}
\newcommand{\OPCAT}{\texttt{OP\_CAT}}
\newcommand{\OPHASH}{\texttt{OP\_HASH160}}
\newcommand{\OPEQUAL}{\texttt{OP\_EQUAL}}
\newcommand{\OPEQUALVERIFY}{\texttt{OP\_EQUALVERIFY}}
\newcommand{\OPCHECKSIG}{\texttt{OP\_CHECKSIG}}

\newcommand{\OPTOALTSTACK}{\texttt{OP\_TOALTSTACK}}
\newcommand{\OPFROMALTSTACK}{\texttt{OP\_FROMALTSTACK}}
\newcommand{\OPCHECKDATASIGVERIFY}{\texttt{OP\_CHECKDATASIGVERIFY}}

\newcommand{\corr}{\ensuremath{\mathsf{Corr}}}
\newcommand{\derkey}{\ensuremath{\mathsf{d}}}
\newcommand{\derkeyi}{\ensuremath{\derkey_i}}

\newcommand{\kws}[0]{Hierarchical Deterministic Wallet; Hierarchical Key Assignment; Bitcoin; Blockchain.}

\ifCONF

  \title{
  Arcula:\\ A Secure Hierarchical Deterministic Wallet\\for Multi-asset Blockchains \thanks{NB: appendices, if any, did not benefit from peer review.\newline A preprint of this paper has been
  deposited on ArXiv.}
  }

\else
    \title{\texorpdfstring{\bf  Arcula:\\ A Secure Hierarchical Deterministic Wallet\\for Multi-asset Blockchains}{Arcula: A Secure Hierarchical Deterministic Wallet for Multi-asset Blockchains}}

    \author{Adriano Di Luzio\thanks{The author is also a PhD Candidate at Sapienza University of Rome, Italy.}}
    \author{Danilo Francati}
    \author{Giuseppe Ateniese}
    \affil{Stevens Institute of Technology, USA\\\texttt{\{adiluzio,dfrancat,gatenies\}@stevens.edu}}
\fi

\begin{document}

\ifCONF{}
    \maketitle
    \begin{abstract}
This work presents Arcula, a new design for hierarchical deterministic wallets that brings identity-based addresses to the blockchain.
Arcula is built on top of provably secure cryptographic primitives.
It generates all its cryptographic secrets from a user-provided seed and enables the derivation of new public keys based on the identities of users, without requiring any secret information.
Unlike other wallets, it achieves all these properties while being secure against privilege escalation.
We formalize the security model of hierarchical deterministic wallets and prove that an attacker compromising an arbitrary number of users within an Arcula wallet cannot escalate his privileges and compromise users higher in the access hierarchy.
Our design works out-of-the-box with any blockchain that enables the verification of signatures on arbitrary messages. We evaluate its usage in a real-world scenario on the Bitcoin Cash network.
\ifFULL
  \vspace{-.3cm}
\fi
\end{abstract}

\else
    \maketitle
    \thispagestyle{empty}
    \begin{abstract}
This work presents Arcula, a new design for hierarchical deterministic wallets that brings identity-based addresses to the blockchain.
Arcula is built on top of provably secure cryptographic primitives.
It generates all its cryptographic secrets from a user-provided seed and enables the derivation of new public keys based on the identities of users, without requiring any secret information.
Unlike other wallets, it achieves all these properties while being secure against privilege escalation.
We formalize the security model of hierarchical deterministic wallets and prove that an attacker compromising an arbitrary number of users within an Arcula wallet cannot escalate his privileges and compromise users higher in the access hierarchy.
Our design works out-of-the-box with any blockchain that enables the verification of signatures on arbitrary messages. We evaluate its usage in a real-world scenario on the Bitcoin Cash network.
\ifFULL
  \vspace{-.3cm}
\fi
\end{abstract}

    \setcounter{page}{1}
\fi

\section{Introduction}\label{sec:Intro}
In recent years, the adoption of blockchain-based crypto-systems grew at an exponential rate.
At their core, these systems create a decentralized and democratic financial world where users exchange their assets without relying on any central authority and where transactions get delivered to the destination in a handful of seconds, by spending cents of a dollar in fees.
For these reasons, these systems attracted the interest of financial operators, banks, (decentralized) exchanges, and e-commerce marketplaces that aim to increase the security and the usability of their systems, reduce their business costs, and prepare for the financial market of the future.
With the adoption of these systems at scale, however, we face a new set of technological and financial hurdles that we have never experienced before:
(1) How to secure the digital assets of a commercial company? (2) How to handle the delegation and separation of responsibilities within the enterprise's chain of command? (3) And how to enable a third-party (\emph{e.g.}, an auditor) to evaluate past transactions of the company, and (4) to gather the list of its assets on the blockchain?

In this work, we aim at solving these challenges, and we focus on hierarchical deterministic wallets.
Similarly to how we keep our coins and our bills in a physical wallet, a blockchain wallet holds all our crypto-assets (\emph{e.g.}, our Bitcoins, Ethers, and other coins).
In particular, wallets hold the cryptographic keys that allow us to spend these coins.
Our goal is to design a wallet that is cryptographically secure and easy to use:
We specifically tackle the use cases of the legacy and future blockchain-based crypto-systems and, in turn, the needs of the large-scale enterprises that rely on them.
The wallet that we aim to design is hierarchical and deterministic.
This means that:
\begin{inlinelist}
\item The users can securely organize their keys in a hierarchy that reflects, for example, the subdivision in departments of a company.
 This way, the managers of a company can spend funds on behalf of their departments, but each department can only spend its own funds.
\item The wallet deterministically generates every cryptographic key by starting from an initial seed (\emph{e.g.}, a pseudorandom sequence derived from a mnemonic phrase).
 As a result, the users can reliably recover all their keys even in the case of total loss (\emph{e.g.}, after a hardware failure or a natural disaster).
\end{inlinelist}

To this end, we present Arcula, a hierarchical deterministic wallet named after the small casket where ancient Romans used to store their jewels.
We build Arcula on a deterministic variation of {\em hierarchical key assignment schemes (HKAs)}, a popular cryptographic primitive that has been studied for more than 20 years.
Arcula implements arbitrarily complex access hierarchies, allows for their dynamic modifications, and can incorporate temporal capabilities into the generation and storage of keys.
In addition, Arcula ties the identities of users to their public signing keys without requiring additional secret information --- bringing to fruition identity-based hierarchical cryptography (\emph{e.g.}, signatures) within the blockchain ecosystem.
Unlike other wallets, Arcula achieves these properties while being formally secure against privilege escalation.
It relies on simple cryptographic primitives and does not depend on any particular digital signature scheme.
As a consequence, Arcula is compatible with any blockchain that enables the verification of signatures on arbitrary messages.
%
%With Arcula, we present a groundbreaking hierarchical deterministic %wallet that will prove itself useful in the immediate time and also in %the future, to enable new applications relying on blockchains and %crypto-coins for which the existing solutions are inadequate.
We show its implementation in action in Bitcoin Cash, a fork of the original Bitcoin crypto-system.

The rest of this work is organized as follows.
In~\Cref{sec:hierarchical_deterministic_wallets} we describe the properties of hierarchical deterministic wallets and their applications.
\Cref{sec:related_works} and~\Cref{sec:preliminaries} respectively discuss the related work and introduce the notation and the cryptographic primitives that we use throughout the paper.
\Cref{sec:security_model_of_hierarchical_deterministic_wallet} formalizes the security model of hierarchical deterministic wallets.
\Cref{sec:arcula} and \Cref{sec:arcula_in_the_real_world} describe our design of Arcula and show it in action in the real world.
Lastly,~\Cref{sec:conclusions} concludes the paper.
\section{Hierarchical Deterministic Wallets}\label{sec:hierarchical_deterministic_wallets}
A hierarchical deterministic wallet (HDW) enables a user to securely generate and store the cryptographic keys associated with her coins.
Blockchain-based crypto-systems typically rely on digital signatures and pairs of secret and public signing keys:
Users spend their assets by signing a transaction with the secret key;
the others verify the authenticity of the signature through the public key.
Public keys can be derived from the corresponding secret keys, but not vice-versa.
On a high level, a hierarchical deterministic wallet holds a collection of secret signing keys.

An HDW deterministically generates its keys by starting from an initial seed provided by the user.
As long as the user remembers the seed, she will be able to recover her keys, even in case of wallet loss.
Besides, an HDW also organizes the keys under an access hierarchy, where each element represents a group of users and associated a pair of signing keys to them.
The privileges of users depend on their level in the hierarchy.
Users with higher privileges (\emph{i.e.}, higher in the hierarchy), must be able to derive the keys of users on lower levels and in turn, to sign messages (\emph{i.e.}, transactions) on their behalf.
Users on lower levels, however, should not be able to escalate their privileges along the hierarchy, not even when colluding with others.
Finally, an HDW should also provide an additional fundamental property:
It must be possible to deterministically generate every public key of the wallet by starting from a \emph{master public key} (depending on the user seed), without requiring any secret information in the process.
As we shall see, this requirement enables a set of creative use cases for HDW (\emph{e.g.}, public auditing of blockchain assets, and generation of new keys in untrusted environments).

\subsection{Properties}\label{sub:properties}
In more detail, let \sski{} be the secret signing key of a user \cvi{} of a hierarchy and let \ppki{} be the corresponding public key.
Let \sskj{} and \ppkj{} be, respectively, the secret and public keys associated with a descendant \cvj{} of \cvi{} (\emph{i.e.}, a user with lower privileges in the hierarchy).
A hierarchical deterministic wallet shall have the following properties:
\begin{property}[Security against privilege escalation]\label{prp:secret_keys}
 For any set of colluding descendants of \cvi{} it is computationally infeasible to recover the secret key \sski{} of \cvi{}.
\end{property}
\begin{property}[Deterministic secret derivation]\label{prp:determinism}
  For each descendant \cvj{} of \cvi{}, the secret keys \sskj{} is deterministically generated by using the secret information of \cvi{}.
  If \cvj{} has the highest privileges in the hierarchy, then her secret information are generated from a user-provided seed.
\end{property}
\begin{property}[Public-key derivation]\label{prp:public_keys}
  The public key \ppkj{} of each descendant \cvj{} of \cvi{} is deterministically generated only using public information;
  the generation process does not require the secret key \sski{} or any other secret information.
\end{property}

The three properties, together, define the ideal hierarchical deterministic wallet.
\Cref{prp:secret_keys} guarantees the security of the wallet:
Colluding users cannot escalate their privileges to recover the secret keys of others higher in the hierarchy.
\Cref{prp:determinism} ensures that all secrets are deterministically generated along the hierarchical path that links the most privileged users to their descendants, and so on.
\Cref{prp:public_keys} requires the public keys to be dynamically derivable without requiring any private information.
As we will see, in the past, this property proved to be particularly hard to achieve while also guaranteeing the security of the underlying wallets.
Nonetheless, it is crucially important as it enables the novel applications of HDW within the blockchain that we discuss in~\Cref{sub:applications}.
Finally, we note that every wallet in which the public derivation property holds, inevitably reveals the relation between the public keys of the wallet, making it impossible to achieve any privacy-related notion, \emph{e.g.}, unlinkability of transactions.
We discuss this issue in~\Cref{sub:transactions_unlinkability}.

% An ideal hierarchical deterministic wallet shall thus guarantee these three properties;
% in addition, the space required by its public parameters should be constant (\emph{i.e.}, should not depend on the size of the wallet), and it should allow the wallet to handle a dynamic (\emph{i.e.}, not predefined) number of keys.
% As we discuss in~\Cref{sub:arcula_at_a_glance}, Arcula achieves all these requirements, but the size of its public parameters grows linearly in the number of keys used to sign transactions.

\subsection{Applications}\label{sub:applications}
Hierarchical deterministic wallets enable different use cases, inspired by both well established and innovative financial applications that specifically tackle the needs of enterprises, governments, and financial institutions.
Individual users will also find HDW useful, leveraging their increased security and their deterministic reliability and, if they choose so, to achieve unlinkability of their transactions.

We discuss the benefits of HDW for individuals in~\Cref{sub:transactions_unlinkability}.
Here, instead, we focus on the applications of HDW at scale, that target (hundreds of) thousands of customers distributed across the world.

\ifCONF{}
 \begin{description}[style=unboxed,leftmargin=0cm]
\else
 \begin{description}
\fi
  \item[Enterprises:]
  In enterprises (\emph{e.g.}, financial institutions, or exchanges), the hierarchy of a deterministic wallet might reflect the underlying chain of command or the subdivision in regions, departments, and teams.
  It allows managers to distribute funds among different branches and ensure fiscal responsibility. In particular, each branch can manage its funds but are not allowed to spend those of other units.
  The deterministic generation of keys simplifies the management of secrets and guarantees a reliable recovery of the wallet, even in case of catastrophic loss.
  Cryptocurrency exchanges that manage the keys of hundreds of thousands of users might find this feature particularly useful:
  Through the HDW, they generate pairs of keys that take into account the hierarchy of users and then rely on the master seed to handle their recovery.
  Finally, the property of public-key derivation enables enterprises to comply with financial laws and regulations without jeopardizing the security of their infrastructure:
  \emph{E.g.} it allows a (possibly untrusted) auditor to inspect the funds that they hold by starting from the master public key and deriving all the public keys in the wallet without relying on any additional secret information.
  \item[e-commerce:]
  An HDW is distinctly beneficial to an e-commerce marketplace.
  Marketplaces, such as Amazon, typically advertise and sell products to buyers. They also allow third-party vendors to do the same.
Cryptocurrencies could help manage the payment flow to these vendors. When selling an item to a buyer, the marketplace generates a fresh payment address for each crypto-coin that it supports.
  As soon as the buyer transfers the required coins to one of the addresses and the blockchain confirms the transaction, the item gets shipped.
  The generation of fresh payment addresses leverages the properties of public-key derivation:
  Since it does not require any private information, it can take place in an untrusted environment (\emph{e.g.}, a web server exposed to the internet) and allows the e-commerce owner to derive the corresponding secret keys only when actually spending the funds (\emph{e.g.}, by deriving them offline starting from an intermediate key of the wallet).
  Even if an attacker compromises the webserver, he will not discover any secret keys, and thus funds received before the attack remain safe and intact. In addition, public-key derivation allows buyers or auditors to check the authenticity of the payment addresses since anyone can generate them from the public key of the marketplace.
 \item[Decentralized Finance (DeFi):]
  Decentralized Finance has recently started replacing many of the existing traditional financial tools (\emph{e.g.}, loans and futures) with open-source alternatives based on the blockchain and its smart contracts.
  With DeFi, HDW can unleash their full potential:
  The execution of smart contracts, indeed, cannot rely on any secret information as both the source code and the processing inputs are stored in the clear on the public blockchain.
  By combining HDW and public derivation, a DeFi smart contract can autonomously derive the public key of the recipient of a transaction (\emph{e.g.}, a user that will receive interests on a loan, or the employee of a company that will receive a percentage of its shares).
  The process could take place on the blockchain and does not require manual interactions.
  In turn, HDW and DeFi lay the foundations of a modern, democratic, and decentralized financial world.
\end{description}

\subsection{Threats and Security Model}\label{sub:threats_and_security_model}
We divide the set of possible threats to the security of hierarchical deterministic wallet in three security levels that we describe according to increasing requirements of trust.
\ifCONF{}
  \begin{description}[style=unboxed,leftmargin=0cm]
\else
  \begin{description}
\fi
  \item[Untrusted Environment: ] This level is entirely untrusted.
    It refers, for example, to the executing environment of a DeFi smart contract that relies on public key derivation (\Cref{prp:public_keys}) to generate fresh addresses to deliver payments.
  \item[Hot Environment: ]
    This level is semi-trusted.
    At this stage, the users can access their own secret keys and derive those of their descendants.
    An attacker that compromises a user of the hot environment will compromise, in turn, only her descendants in the hierarchy (\Cref{prp:secret_keys}).
  \item[Cold Storage: ]
    This is the most trusted level of the security model, holding the seed used to generate every key within the wallet deterministically.
    It typically corresponds to an offline location (\emph{e.g.}, a hardware token used to instantiate the wallet) that is physically secured (\emph{e.g.}, in a safe).
    An attacker that compromises the cold storage has full access to the wallet and to every asset that it holds.
\end{description}
In~\Cref{sub:arcula_at_a_glance,sec:related_works}, we analyze our construction and the related work under the spotlight of this security model.

\subsection{Arcula at First Glance}\label{sub:arcula_at_a_glance}

\begin{figure*}[tb]
  \centering
  \includegraphics[width=0.8\linewidth]{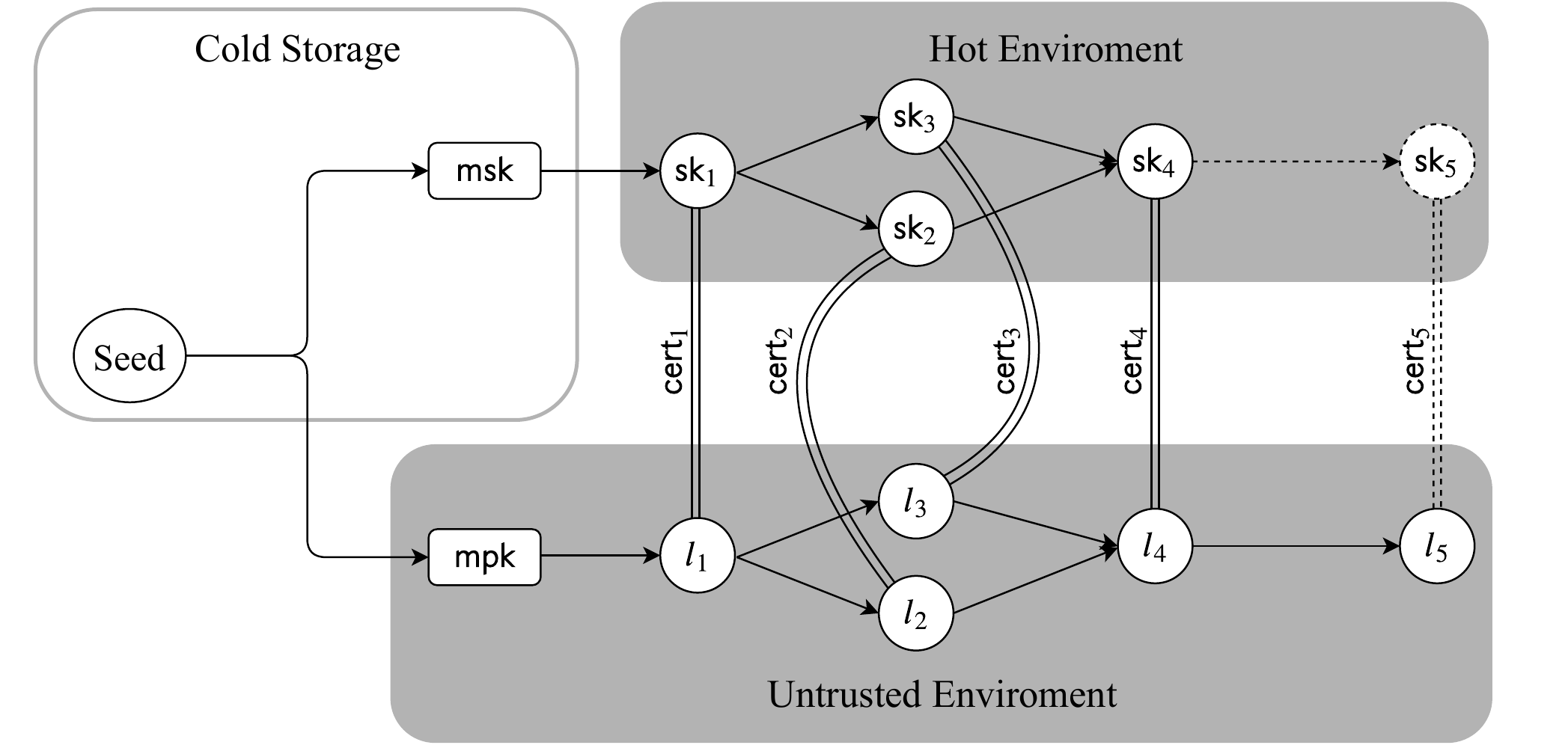}
  \caption{%
    A glance at the deterministic generation of secrets and identity-based public key derivation within Arcula.
    The users $l_1$ to $l_5$ at the bottom of the figure create a hierarchy, encoded as a directed acyclic graph.
    Arcula starts by deriving a pair of master keys $(\msk, \mpk)$ from the seed.
    The users are unequivocally identified by concatenating the master public key \mpk{} with their identities (\emph{e.g.}, $l_5$).
    The master secret key \msk{} allows the deterministic generation of secrets, so that every secret transitively depends on the initial seed.
    The hierarchical deterministic key assignment at the core of Arcula relies on the \msk{} and enables users with higher privileges to derive the secrets of their descendants (\emph{e.g.}, enabling $l_1$ to start from her secret information and derive the secret key $\sk{}_4$ of $l_4$).
    Finally, Arcula associates secret keys to their corresponding identities by providing each user with a certificate signed by the \msk{}:
    $\cert_{4}$, \emph{e.g.}, associates $l_4$ with the secret key $\sk_{4}$.
    This approach enables Arcula to flexibly add new users to the hierarchy without requiring secret information.
    The user $l_5$ was dynamically added to the hierarchy, and her public key was obtained through identity-based derivation, but she still does not have any secret information associated with her.
    This means that she can receive funds (\emph{e.g.}, on the blockchain) and that the certificate and her secret information will only be required when she will spend them.
  }\label{fig:resources/arcula_glance}
\end{figure*}

\Cref{fig:resources/arcula_glance} provides an overview of the design of Arcula, our hierarchical deterministic wallet.
Arcula starts from the seed to deterministically generate a master pair of secret and public keys $(\msk, \mpk)$.
After the initial instantiation, the seed and the master secret key can be safely moved to cold storage.
The master public key uniquely identifies the wallet and will be used in the public key derivation process.
In more detail, Arcula generates the public key of $i$-th user of the hierarchy by concatenating the master public key $\mpk$ and her identity $\cli$ (\emph{e.g.}, a numerical index or a bit string).
As a result, public-key derivation in Arcula can be executed, by design, in any untrusted environment.
The master secret key $\msk$ allows, instead, to deterministically generate the secret keys $\sski$ corresponding to the users of the hierarchy.
In particular, Arcula generates the secret keys by relying on deterministic hierarchical key assignment:
A provably secure cryptographic scheme through which we assign a derivation key to every user of the hierarchy, that, in turn, they will use to derive their own secret key and those of their descendants.
The private key generation should be executed in the context of the hot environment, as compromising a single user will lead to compromising every descendant.
Finally, Arcula explicitly associates secret keys to their corresponding identity-based public keys through a certificate $\cert_i$, signed by the master secret key $\msk$, that links the public identity $\cli$ (top of~\Cref{fig:resources/arcula_glance}) to the secret key \sski{} generated by the deterministic key assignment scheme (bottom of~\Cref{fig:resources/arcula_glance}).

The identity-based approach that we realized with Arcula provides several advantages:
First, it solves the problem of distributing a public key for each user of the hierarchy and of associating, off the blockchain, to the identity of an individual.
In addition, it allows for generating an unbounded number of addresses for receiving transactions.
On the other hand, Arcula relies on certificates signed by the master secret key to associate the secret key of a user with their identity, to which the transaction was addressed.
The users only require these certificates when they sign a transaction for the first time:
This means that their creation can be delayed until that moment and that it can happen entirely offline.
% Nonetheless, it also means that every new certificate will be stored in the public parameters of the wallet and, in practice, on the blockchain (their size is typically as small as a single group element in signature schemes based on Elliptic Curves, \emph{e.g.}, ECDSA).

\section{Related Work}\label{sec:related_works}

\ifCONF
\begin{table*}[tb]
\else
\begin{table}[tb]
\small
\fi
  \centering
  \caption{Comparison between Arcula and the existing state-of-the-art solutions.}\label{tab:state_of_the_art}
  \begin{tabular}{c c c c c}
    \toprule
\ifCONF
  & \begin{tabular}{c} Security to Privilege Escalation \\ (\Cref{prp:secret_keys}) \end{tabular} &
  \begin{tabular}{c} Public Key Derivation \\ (\Cref{prp:public_keys}) \end{tabular} &
  \begin{tabular}{c} Deterministic Generation \\ (\Cref{prp:determinism}) \end{tabular} &
\else
  & \begin{tabular}{c} Security to \\ Privilege Escalation \\ (\Cref{prp:secret_keys}) \end{tabular} &
  \begin{tabular}{c} Public Key \\ Derivation \\ (\Cref{prp:public_keys}) \end{tabular} &
  \begin{tabular}{c} Deterministic \\ Generation \\ (\Cref{prp:determinism}) \end{tabular} &
\fi
  Hierarchy \\
  \midrule
  BIP32~\cite{wuille2012bip32} & No & Yes & Yes & Tree \\
  Hardened BIP32~\cite{wuille2012bip32} & Yes & No & Yes & Tree \\
  Gutoski and Stebila~\cite{gutoski2015} & No & Bounded & Yes & Tree \\
  Fan \emph{et al.}~\cite{dblp:conf/desec/fantshk18} & No & Yes & Yes & Tree \\
  \midrule
  Poulami \emph{et al.}~\cite{das2019formal} & $-$ & $-$ & $-$ & No \\
  Goldfeder \emph{et al.}~\cite{goldfeder2015securing} & $-$ & $-$ & $-$ & No \\
  Gennaro \emph{et al.}~\cite{gennaro2016} & $-$ & $-$ & $-$ & No \\
  Dikshit and Singh~\cite{dikshit2017} & $-$ & $-$ & $-$ & No \\
  \midrule
  Arcula (\Cref{sec:arcula}) & Yes & Yes & Bounded & DAG \\
  \bottomrule
  \end{tabular}
\ifCONF
\end{table*}
\else
\end{table}
\fi

Bitcoin Improvement Proposal 32 (BIP32) defines the state of the art implementation of hierarchical deterministic wallets~\cite{wuille2012bip32}.
In short, let \cGroupg{} be the generator point of an Elliptic Curve.
A private key \sski{} is associated with its public key \ckipub{}.
Let \hash{} be a hash function; the descendants' private keys \sskj{} is defined as:
\begin{equation}\label{eq:secret_key}
  \sskj{} = \hash(\ppki{} \concat j) + \sski{}
\end{equation}
The corresponding public keys \ppkj{}, instead:
\ifFULL{}
\begin{align}
  \ppkj{} &= \cGroupPub{\sskj{}} \nonumber \\
  \ppkj{} &= \cGroupPub{\hash(\ppki{} \concat j) + \sski{}} \nonumber \\
  \ppkj{} &= \cGroupPub{\hash(\ppki{} \concat j)} \cdot \cGroupPub{\sski{}} \nonumber \\
  \ppkj{} &= \cGroupPub{\hash(\ppki{} \concat j)} \cdot \ppki{}\label{eq:public_key}
\end{align}
\else
\begin{equation}
  \ppkj{} =
  \cGroupPub{\sskj{}} =
  \cGroupPub{\hash(\ppki{} \concat j) + \sski{}} =
  % \cGroupPub{\hash(\ppki{} \concat j)} \cdot \cGroupPub{\sski{}} =
  \cGroupPub{\hash(\ppki{} \concat j)} \cdot \ppki{}\label{eq:public_key}
\end{equation}
\fi
\Cref{eq:secret_key,eq:public_key} satisfy the properties of deterministic generation and public derivation (\Cref{prp:determinism,prp:public_keys}).
However,~\Cref{eq:secret_key} creates a privilege escalation vulnerability where the knowledge of a descendant private key \sskj{} and the parent public key \ppki{} allows recovering the parent private key \sski{}:
\begin{equation*}
  \sski{} = \sskj{} - \hash(\ppki{} \concat j) \mod \cGroupq
\end{equation*}
This privilege escalation vulnerability has been discussed extensively~\cite{wuille2012bip32,vitalik:2013,courtois2014privatekr,gutoski2015}.
In the context of our threat model, there is no distinction between the cold storage and the hot environment, since compromising any node leads to compromising the entire wallet.

BIP32 addresses this issue by designing a \emph{hardened} key derivation method that generates a descendant private key $\sk^{h}_{j}$ as follows:
\begin{equation}
  \sk^{h}_{j} = \hash(\sski{} \concat i) + \sski{} \mod \cGroupq{}\label{eq:secret_key_h}
\end{equation}
The hardened derivation solves the privilege escalation vulnerability but looses the public key derivation (\emph{i.e.}, trades~\Cref{prp:public_keys} for~\Cref{prp:secret_keys}).
Generating a hardened public key $\pk^{h}_{j}$ now requires the parent secret key \sski{}:
\begin{equation*}
  \pk^{h}_{j} = \cGroupPub{\sk^{h}_{j}} = \cGroupPub{\hash(\sski \concat j) + \sski{}} = \cGroupPub{\hash(\sski{} \concat j)} \cdot \ppki{}
\end{equation*}

In~\Cref{tab:state_of_the_art} we compare Arcula, our hierarchical deterministic wallet, with (hardened) BIP32 and the related works.
BIP32 does not satisfy the security to privilege escalation property.
The hardened version of BIP32, instead, fails in deriving public keys without requiring additional secrets.
Gutoski and Stebila~\cite{gutoski2015} propose an HDW strengthens the security of BIP32.
Their design splits each secret key into $n$ shares, distributed to the descendants of the user;
reconstructing the secret key requires at least $m$ shares.
This solution provides weaker security than~\Cref{prp:secret_keys}, because $m$ colluding descendants of a user can recover the original secret key (as opposed to preventing any set of colluding descendants from escalating their privileges).
In addition, they support~\Cref{prp:public_keys} by publishing the public keys of all the users in the wallet.
They do not allow the generation of fresh public keys, and their derivation is bounded to the number of published keys.
Fan \emph{et al.}~\cite{dblp:conf/desec/fantshk18} develop an HDW based on Schnorr signatures and trapdoor hash functions that enables the users to sign new transactions without accessing their private keys.
A generic user can sign transactions on behalf of its descendants only after authorization by the root of the hierarchy that needs to reveal her the master private trapdoor key.
As a result, any authorized user is able not only to sign new transactions on behalf of its descendants but also of all the users of the hierarchy.
Compromising a single authorized user leads to revealing every secret stored in the wallet --- the cold storage and hot environment of our threat model overlap, and the scheme is not secure against privilege escalation.
Poulami \emph{et al.}~\cite{das2019formal} provide a formal definition of non-hierarchical deterministic wallets and show a set of modifications that make ECDSA-based deterministic wallets provably secure.
Goldfeder \emph{et al.}~\cite{goldfeder2015securing} and subsequently Gennaro \emph{et al.}~\cite{gennaro2016} propose a non-hierarchical deterministic wallet where the secret key is shared among $n$ parties, and at least $t$ of them are required to sign a transaction.
Dikshit and Singh~\cite{dikshit2017} extend the threshold-based ECDSA signatures to assign different weights to the participants of the protocol.
These works deal with non-hierarchical deterministic wallets, and they do not aim at achieving~\Cref{prp:secret_keys,prp:determinism,prp:public_keys} (depending on the hierarchical structure of the wallet).

Arcula, on the other hand, is the only solution secure against privilege escalation that, at the same time, enables identity-based unbound public key derivation.
Spending the coins addressed to one of its users, however, requires the creation of a certificate that associates their identities to their keys.
For this reason,~\Cref{prp:determinism} and, more precisely, spending coins within Arcula, is bound to the generation of a certificate, signed by the master secret key, that authorizes the users.
The deterministic generation of secret keys, instead, relies on an unbound hierarchical key assignment.
This approach provides some significant advantages and enables a set of novel use cases:
Users (\emph{e.g.}, enterprises, or smart contracts in decentralized finance) rely on the identity-based derivation to generate in untrusted environments fresh public keys and addresses on which they will securely receive coins.
The certificates will only be required before spending these coins, and their generation can happen entirely offline (\emph{e.g.}, in cold storage).

Finally, Arcula differentiates itself from the other solutions by supporting a complex access hierarchy --- \emph{e.g.}, a directed acyclic graph (DAG) --- instead of a tree.
Many of our novel use cases require more complex hierarchies.
Consider, for instance, two or more departments of a company that collaborate on a project.
They share a common budget and need to spend from it without revealing their own secret key.
The access hierarchy encoding this context would derive a single node, the budget, as a successor of multiple others (the departments).
\Cref{eq:secret_key,eq:public_key,eq:secret_key_h} cannot handle hierarchies where a node has more than one predecessor---constraining BIP32 only to implement access hierarchy that forms a tree.
The hierarchical key assignment scheme at the core of Arcula, instead, supports complex hierarchies and, in addition, it allows us to dynamically modify the hierarchy (\emph{e.g.}, by deleting an intermediate node) and to incorporate temporal capabilities (\emph{e.g.}, key expiration) into the wallet.

\section{Preliminaries}\label{sec:preliminaries}

\subsection{Notation}\label{sub:notation}
We use the notation $[n] = \{1,\ldots,n\}$.
Uppercase boldface letters (such as $\rv{X}$) are used to denote random variables, lowercase letters (such as $x$) to denote concrete values, calligraphic letters (such as $\cXX$) to denote sets, and sans serif letters (such as \Adv{}) to denote algorithms.
Algorithms are modeled as (possibly interactive) Turing machines;
if algorithm \Adv{} has access to some oracle \Oracle{}, we often write $\cQQ_{\Oracle{}}$ for the set of queries asked by \Adv{} to \Oracle{}.

For a string $x\in\bin^*$, we let $|x|$ be its length;
$|\cXX|$ represents the cardinality of the set $\cXX$.
When $x$ is chosen randomly in $\cXX$, we write $x \sample \cXX$.
We write $y = \Adv(x)$ to denote a run of the algorithm \Adv{} on input $x$ and output $y$;
if \Adv{} is randomized, $y$ is a random variable and $\Adv{}(x;r)$ denotes a run of \Adv{} on input $x$ and (uniform) randomness $r$.
We write $y \sample \Adv(x)$ to denote a run of the randomized algorithm \Adv{} over the input $x$ and uniform randomness.
An algorithm \Adv{} is \emph{probabilistic polynomial-time} (PPT) if \Adv{} is randomized and for any input $x,r\in\bin^*$ the computation of $\Adv(x;r)$ terminates in a polynomial number of steps (in the input size).

Throughout the paper, we denote by $\secpar\in\NN$ the security parameter and we implicitly assume that every algorithm takes as input the security parameter. A function $\nu:\NN\rightarrow[0,1]$ is called \emph{negligible} in the security parameter $\secpar$ if it vanishes faster than the inverse of any polynomial in $\secpar$, \emph{i.e.}, $\nu(\secpar)\in \bigO{1/p(\secpar)}$ for all positive polynomials $p(\secpar)$.
We write $\negl$ to denote an unspecified negligible function in the security parameter.

\subsection{Signature Scheme}\label{subsec:signature-scheme}
A signature scheme with message space $\cMM$ is made of the following polynomial-time algorithms.
\begin{description}
  \item[$\Keygen(\secparam)$:] The randomized key generation algorithm takes the security parameter and outputs a secret and a public key $(\sk, \pk)$.

  \item[$\Sign(\sk, m)$:] The randomized signing algorithm takes as input the secret key $\sk$ and a message $m\in\cMM$, and produces a signature $\sigma$.

  \item[$\Verify(\pk, m,\sigma)$:] The deterministic verification algorithm takes as input the public key $\pk$, a message $m$, and a signature $\sigma$, and it returns a decision bit.
\end{description}

A signature scheme is correct if honestly generated signatures always verify correctly.
\begin{definition}[Correctness of signatures]\label{def:Sigcorrectness}
    A signature scheme $\Pi= (\Keygen, \Sign, \allowbreak{} \Verify)$ with message space $\cMM$ is correct if $\forall\secpar\in\NN$ and $\forall m \in \cMM$, the following holds:
    \[
    \prob{\Verify(\pk, m,\Sign(\sk, m))}=1,
    \]
    where $(\sk, \pk)\sample \Keygen(\allowbreak\secparam)$.
\end{definition}

For security we are interested in existential unforgeability, \emph{i.e,} it must be infeasible to forge a valid signature on a new fresh message.

\begin{definition}[Unforgeability of signatures] \label{def:SigUnforgeability}
	A signature scheme $\Pi= (\Keygen, \allowbreak \Sign, \Verify)$ is existentially unforgeable under chosen-message attacks if for all PPT adversaries $\Adv{}$:
	\[
	\prob{\GameSigForge = 1} \leq \negl,
	\]
	where $\GameSigForge$ is the following experiment:
	\begin{description}
		\item[Setup:] The challenger runs $(\sk, \pk) \sample\Keygen(\secparam)$ and gives $\pk$ to \Adv{}.
		\item[Query:] The adversary has access to a signing oracle $\SignOracle{}(\cdot)$. On input $m$, the challenger computes and returns $\sigma \sample \Sign(\sk, m)$. Let $\cQQ_{\Sign}$ denote the the messages queried to the signing oracle.
		\item[Forgery:] The adversary outputs $(m, \sigma)$. If $m \not\in \cQQ_\Sign$, and $\Verify(\pk, \allowbreak m, \sigma) = 1$, output $1$, else output $0$.
	\end{description}
\end{definition}

\section{Security Model of Hierarchical Deterministic Wallet}\label{sec:security_model_of_hierarchical_deterministic_wallet}
One of the contributions of this work is to formalize, for the first time, the security model of hierarchical deterministic wallets (HDW).
We define a hierarchical deterministic wallet over $5$ algorithms \SetDerivePPSignVerify{}, where:
\begin{inlinelist}
\item $\Set{}$ deterministically instantiates the wallet by generating the public parameters \pp{} and a set of the derivation key \derkeyi{}, one for each node $\cvi\in\cV$ of the hierarchy.
\item $\DerivePriv{}$ and $\DerivePub$ are responsible of the derivation of signing and public keys (\Cref{prp:determinism,prp:public_keys}).
  $\DerivePriv{}$ derives the signing key \sskj{} of a node \cvj{}, descendent of \cvi{}, by using the derivation key \derkeyi{} associated to \cvi{};
  $\DerivePub{}$ derives the corresponding public key \ppkj{} by using the only the public parameters \pp{}.
\item $\Sign$ and $\Verify$ take inspiration from the standard signing and verification algorithms of a digital signature scheme.
\end{inlinelist}

Concretely, let \sG{} be a directed acyclic graph (DAG) representing an access hierarchy.
We define the set of descendants \ddesc{} of node \cvj{} to be the set of nodes \cvj{} such that there exists a direct path $w$ from \cvi{} to \cvj{} in \cG{}.
A hierarchical deterministic wallet $\Pi=\SetDerivePPSignVerify{}$, defined over a seed space \cSS{} and message space \cMM{} is defined in the following way:
\begin{description}
  \item[\sdSet{}:] The deterministic setup algorithm takes as input a security parameter, an access graph \sG{}, and an initial seed $\cSeed \in \cSS{}$, and outputs the public parameters \pp{} and a set of derivation keys $\{\derkeyi{}\}_{\cvi \in \cV}$.
  \item[\swDerivePub{}:] The deterministic public derivation algorithm takes as input the public parameters \pp{}, a target node \cvi{}, and outputs the public key \ppki{} associated to node \cvi{}.
  \item[\swDerivePriv{}:] The deterministic private derivation algorithm takes as input the public parameters \pp{}, the derivation key \derkeyi{} of node \cvi{}, and a target node $\cvj{} \in \descvi$, and outputs the secret key \sskj{} associated to node \cvj{}.
  \item[\sSign{}:] The randomized signing algorithm takes as input a message $m \in \cMM$, and a secret key \sski{}, and outputs a signature $\sigma$.
  \item[\sVerify{}:] The deterministic verification algorithm takes as input a public key \ppki{}, a message $m$, and a signature $\sigma$, and outputs a decisional bit $b$.
\end{description}

A hierarchical deterministic wallet is correct if any user can derive the private and public key of its descendants and create a valid signature on behalf of them.
This means that any node \cvi{} can derive the signing key \sskj{} of any node $\cvj{} \in \desc{}(\cvi)$ and produce, in turn, a valid signature $\sigma$ on behalf of \cvj{} (\emph{i.e.}, that passes the verification process against the public key \ppkj{} obtained through public key derivation).
\begin{definition}[Correctness of HDW]\label{def:HDW-corr}
A hierarchical deterministic wallet $\Pi = \SetDerivePPSignVerify{}$, with seed space \cSS{} and message space \cMM{}, is correct if for every DAG $\cG = (\cV, \cE)$, $\forall \cvi,\cvj \in V, \, \forall \cvj \in \desc{}(\cvi), \forall \cSeed \in \cSS, \forall m \in \cMM$ the following conditions holds:
\begin{equation*}
  \prob{\sVerifyJ{}= 1} \geq 1 - \negl{},
\end{equation*}
where $\ppssppkk = \sdSet$, $\sskj = \swDerivePriv$, and $\ppkj = \DerivePub(\pp{},\allowbreak \cvj{})$.
\end{definition}

The security of a hierarchical deterministic wallet draws inspiration from existentially unforgeable signatures.
We allow an attacker to corrupt an arbitrary number of users in the hierarchy---by corrupting a user; the attacker implicitly corrupts also all her descendants.
In addition, the attacker also has access to a signing oracle, that returns signatures on arbitrary messages from any uncorrupted node.
We challenge the attacker to forge a signature for a new message on behalf of an uncorrupted node.
\begin{definition}[Hierarchical existential unforgeability of HDW]\label{def:HDH-euf}
  A hierarchical deterministic wallet is hierarchically existentially unforgeable under chosen-message attacks if for every DAG $\cG=(\cV, \cE)$ and PPT adversary \Adv{} the following condition holds:
  \begin{equation*}
   \prob{\GameHSigForge = 1} \leq \negl,
  \end{equation*}
  where experiment \GameHSigForge{} is defined in the following way:
  \begin{description}
    \item[Setup:] The challenger samples a random $\cSeed \sample \cSS$ and executes $\ppssppkk = \sdSet$.
    It gives the public parameters \pp{} to \Adv{}.
    \item[Query:] The adversary \Adv{} has access to the following oracles:
      \begin{description}
        \item[$\CorrOracle(\cdot)$:] On input $\cvi{} \in \cV$, the challenger answers by giving \derkeyi{} to \Adv{}.
          Let $\cQQ_{\corr}$ denote the set of nodes \cvi{} that \Adv{} corrupted, including their descendants \descvi{}.
        \item[$\SignOracle(\cdot, \cdot)$:] On input $(m, \cvi{}) \in \cMM \times \cV$, the challenger returns \ssigma{} where \ssDerivePrivArculaZero{}.
          Let $\cQQ{}_\Sign$ denote the pairs $(m, \cvi{})$ for which \Adv{} queried the oracle $\SignOracle$.
      \end{description}
    \item[Forgery:] \Adv{} outputs a forgery $(\cvi{}, m, \sigma)$.
     If $\Verify_{\ppki}(m, \sigma) = 1$ where $\ppki{} = \swDerivePub{}$ and $\cvi \notin \cQQ_\corr$, $(m, \cvi{}) \notin \cQQ_\Sign$, return $1$; otherwise return $0$.
  \end{description}
\end{definition}

\section{Arcula: A Secure Hierarchical Deterministic Wallet}\label{sec:arcula}
Arcula, the design that we present in this paper, satisfies the properties of HDW formalized by~\Cref{def:HDW-corr,def:HDH-euf}.
It is provably secure against key recovery;
it deterministically derives the private information from an initial seed;
and it enables identity-based public key derivation.
It achieves such result by relying on two fundamental intuitions:
Securely and deterministically generating a set of keys for the users of a hierarchy and then explicitly associating these keys with their identities.
By doing so, Arcula provides a groundbreaking design that brings identity-based cryptography to the blockchain, tackles explicitly novel use cases and applications, and stands on more than $20$ years of research in how to securely distribute keys to the users in a hierarchy.

At its core, indeed, Arcula derives the keys of the users by relying on a deterministic version of hierarchical key assignment schemes (HKA):
A provably secure process that, precisely as it happens in an HDW, takes as input a hierarchy of users and assigns a secret key to every user, so that users with higher privileges can derive the keys of those with fewer privileges.
To the best of our knowledge, hierarchical key assignment schemes have never been leveraged before implementing an HDW\@.
As a result, one of the contributions of this work is to bind together, for the first time, these seemingly unrelated fields of research.
In addition, hierarchical key assignment schemes provide several advantages:
They're highly efficient and have been extensively studied in the past;
they enable Arcula to implement arbitrarily complex hierarchies, to integrate temporal capabilities into the wallet, and to support the dynamic addition or removal of users to the hierarchy.

The following sections first provide a brief background on (deterministic) hierarchical key assignment schemes and then describe, in detail, our construction of Arcula.

\subsection{Deterministic Hierarchical Key Assignment}\label{sub:deterministic_hierarchical_key_assignment}
A hierarchical key assignment scheme~\cite{atallah:2009:dek:1455526.1455531} assigns a set of cryptographic keys to a set of users in a hierarchy.
The hierarchy, encoded as a directed acyclic graph, represents the access rights of users:
A path from a node \cvi{} to a node \cvj{} implies that the user \cvi{} has higher privileges than \cvj{} and can assume the same access rights of \cvj{}.
An efficient hierarchical key assignment scheme (HKA) enforces the access hierarchy while minimizing the number of keys distributed to the users.

Typically, HKA schemes sample the cryptographic secrets that they assign at random.
Our goal, however, is to leverage an HKA at the core of our wallet, where each secret is deterministically derived from a seed provided by the user.
To do so, we propose a deterministic modification of the HKA developed by Atallah \emph{et al.}~\cite{atallah:2009:dek:1455526.1455531} that is secure under key indistinguishability (and as a consequence guarantees~\Cref{prp:secret_keys} by design).

A Deterministic Hierarchical Key Assignment (DHKA) scheme with seed space $\cSS$ is composed of the following polynomial-time algorithms:
\begin{description}
  \item[\sdSet{}:] The deterministic setup algorithm takes as input the security parameter, a DAG \sG{}, and an initial seed \sSeed{}, and outputs two mappings:
    \begin{inlinelist}
    \item a public mapping \sPub{}, associating a public label \cli{} to each node \cvi{} in \cG{} and a public information \cyij{} to each edge \sE{};
    \item a secret mapping \sSec{}, associating a secret information \cSi{} and a cryptographic key \cxi{} to each node \cvi{} in \cG{}. (No secret information is associated to the edges).
    \end{inlinelist}
  \item[\sDerive{}:]\tDerive{}
\end{description}

Informally, the correctness of a DHKA scheme requires that every user must be able to derive, correctly, the secret key of any other user lower in the hierarchy.
Its security definition, instead, requires that even if an attacker corrupts an arbitrary number of descendants of a node, he cannot distinguish its secret key from a uniformly random string.
We formalize the security model and the construction of our DHKA in~\Cref{sec:deterministic_wallets_through_dynamic_and_efficient_key_management}.

\subsection{Constructing Arcula from DHKA and signatures}\label{sub:arcula_implementation}

Arcula, our implementation of a hierarchical deterministic wallet, relies on deterministic hierarchical key assignment schemes (DHKA) and digital signatures.
This section details our construction, that we describe through the \SetDerivePPSignVerify{} algorithms that define a hierarchical deterministic wallet.
The \Set{} algorithm instantiates the scheme, starting from the initial seed to deterministically generate a pair of \emph{master public and secret} keys \mspk{}.
The master keys respectively serve two purposes:
To identify the users of the wallet and to assign a pair of \emph{signing} keys \ssppkifirst{} to each of them.
We generate the signing keys through the key indistinguishable DHKA scheme of~\Cref{sub:deterministic_hierarchical_key_assignment}, so that users can sign transactions on behalf of their descendants while guaranteeing the security of their keys against privilege escalation.
In more detail, the DHKA assigns to each user \cvi{} a derivation key \derkeyi{} that allows them to generate their own signing keys \skXi{} and to derivate those of their descendants.
We set the master secret key to the derivation key of the user \cvzero{} with highest privileges in the hierarchy\footnote{If there exist multiple users with maximum privileges it is always possible to modify the structure of the DAG and add a \emph{minimal} node that does not correspond to any real user, does not change the hierarchical ordering of the others, and has the highest privileges. We describe the modification in detail in~\Cref{sec:deterministic_wallets_through_dynamic_and_efficient_key_management}.}, \emph{i.e.}, $\msk = \derkey_0$.
The master public key \mpk{}, on the other hand, unequivocally identifies the wallet.
We set it to the public signing key of \cvzero{} ($\mpk = \pkX_0$), and we combine it with the identifiers of users to achieve public key derivation.
In particular, we identify a user \cvi{} of the wallet by concatenating the master public key \mpk{} with the public label \cli{} associated with her.
Finally, the \Set{} algorithm binds the identifies of users to their signing key pair \ssppkifirst{} through a certificate, signed by the master secret key \msk{}\footnote{The master secret key \msk{} deterministically generates the (master) signing key $\skX_0{}$, that we use to sign the certificates.}, that explicitly authorizes the signing key \skXi{} to spend the coins destined to \cvi{}.
The \DerivePriv{} and \DerivePub{} algorithms respectively derive the signing keys (through the DHKA scheme) and the corresponding public keys (by combining the master public key \mpk{} and the node identifiers).
Finally, the \Sign{} and the \Verify{} algorithms handle the creation and verification of digital signatures.
Any node \cvi{} runs the \Sign{} algorithm with its signing key \skXi{} to create a signature and the \Verify{} algorithm checks that there exists a certificate authorizing \skXi{} to spend funds on behalf of the node \cvi{} (identified by the label \cli{}) under the master secret key \msk{}.

On the blockchain, the master public key \mpk{} and the public label \cli{} of a user form her address for receiving payments.
The users spend funds by signing transactions through their signing key \skXi{} and by presenting, at the same time, the certificate that associates their keys to their identity and that authorizes them to spend funds.
This approach provides several advantages.
Receiving funds does not require the certificate---the identity-based public key derivation allows us to generate the address of any destination node without any private information.
Users only require their certificate when spending funds for the first time, \emph{i.e.}, when signing a transaction through their secret key \skXi{}, and the creation of the certificate can happen entirely offline (\emph{i.e.}, in cold storage).
In addition, the signing keys \ssppkifirst{} can also provide users with unlinkability of their transactions.
As we describe in~\Cref{sub:transactions_unlinkability}, users that do not need the identity-based public key derivation can directly leverage the signing public key \pkXi{} as a pseudonym address on which to receive funds and then spend them through the corresponding private key \skXi{}.

\begin{construction}\label{const:cert-arcula}
Let \sDHKA{} and \dECDSA{} be respectively a DHKA and a signatures signature scheme.
We build Arcula in the following way:
\begin{description}
 \item[\sdSet:] On input the security parameter, a DAG $\cG{} = (\cV, \cE)$, and a seed $\cSeed \in \cSS$ the algorithm proceeds as follows:
 \begin{enumerate}
   \item Compute $\PubSec{} = \sdSetDHKA{}$.
   \item For each node $\cvi \in \cV$:\label{enum:w_set_proof}
     \begin{enumerate}
       \item Let $\csxi{} = \aSec{}$ and set $\derkeyi{} = \cSi{}$.\label{enum:w_set_proof_skip}
       \item $(\skX_i, \pkX_i) = \Keygen_{\Sig}(\secparam; \cxi{})$.\label{enum:w_set_proof_inner}
     \end{enumerate}
   \item Output $\pp{} = (\cG{}, \Pub{}, \crtset{}, \pkX_0)$ and ${\{\derkeyi\}}_{\cvi\in\cV}$ where $\hatsigmai \sample \Sign_{\Sig}(\skX_0, (\pkX_i, \cli{}))$ for $\cvi \in \cV$, and $\cli = \Pub(\cvi)$.\label{enum:sign_oracle_proof_inner}
 \end{enumerate}

 \item[$\DerivePub(\pp{}, \cvj)$:] On input the public parameters $\pp{} = (\cG{}, \Pub{}, \crtset{}, \pkX_0)$ and a node $\cvj{} \in \cV$, the algorithm returns $\pk_j= (\pkX_0, \clj{})$ where $\clj = \Pub(\cvj)$.

 \item[$\DerivePriv(\pp{}, \derkeyi, \cvi, \cvj)$:] On input the public parameters $\pp{} = (\cG{}, \Pub{}, \crtset{}, \pkX_0)$, the derivation key $\derkeyi{} = \cSi{}$, and two nodes $\cvi{}, \cvj{}\in \cV$ such that $\cvj{} \in \descvi$, the algorithm runs $\cxj = \DeriveDHKA(\cG{}, \Pub{}, \cvi{}, \cvj{}, \cSi{})$ and $(\skX_j, \pkX_j) = \Keygen_{\Sig}(\secparam; \cxj{})$.
 Finally, it returns $\sk_j = (\skX_j, \pkX_j, \hatsigma_j)$.

 \item[$\Sign(\sk_{i}, m)$:] On input a signing key $\sk_i = (\skX_i, \pkX_i, \hatsigma_i)$ and a message $m$, the algorithms returns $\sigma = (\pkX_i, \sigma', \hatsigmai)$ where $\sigma' \sample \Sign_{\Sig}(\skX_i, m)$.

 \item[$\Verify(\mathsf{pk}_{i}, m, \sigma{})$:] On input a public key $\pk_i=(\pkX_0,\cli)$, a message $m$, and a signature $\sigma{} = (\pkX_i, \sigma', \hatsigmai)$, the algorithms returns $1$ if $\Verify_{\Sig}(\pkX_0, (\pkX_i, \cli{}), \hatsigmai) = 1$ and $\Verify_{\Sig}(\pkX_i, m, \sigma') = 1$; otherwise it returns $0$.
\end{description}
\end{construction}

\begin{remark}
  Arcula shares a significant number of similarities with identity-based hierarchical signature schemes~\cite{10.1007/3-540-36178-2_34}.
  Its design, indeed, associates a pair of signing keys to the identity of a user and allows her to sign messages on behalf of her descendants.
  We point out, however, some fundamental differences.
  Most hierarchical identity-based signature schemes leverage a conspicuous number of public parameters and rely on bilinear mappings.
  This makes them unpractical to use in the existing blockchains:
  The underlying protocols should efficiently handle the bilinear mappings, and the public parameters that define the instantiation of the scheme should be stored on the blockchain itself.
  Our design of Arcula, on the other hand, explicitly takes into consideration the characteristics and the limitations of blockchain systems:
  We do not rely on bilinear mappings, and we only store a small portion of the public parameters \pp{} (one certificate per transaction, typically a single group element) on the blockchain.
\end{remark}

The correctness of the scheme comes directly from the correctness of the underlying primitives.
As for security, we establish the following result whose proof appears in~\Cref{subsec:proof_arcula}.
\begin{theorem}\label{thm:arcula_euf}
  Let \sDHKA{} and $\Sig = (\kgen_\Sig, \Sign_\Sig, \Verify_\Sig)$ be respectively a deterministic hierarchical key assignment and a signature scheme.
  If \DHKA{} is key indistinguishable (\Cref{def:sec_dka}) and $\Sig$ is existentially unforgeable (\Cref{def:SigUnforgeability}), then the HDW $\Pi$ from \Cref{const:cert-arcula} is hierarchically existentially unforgeable (\Cref{def:HDH-euf}).
\end{theorem}

In the context of the security model that we defined in~\Cref{sub:threats_and_security_model}, Arcula's public derivation belongs to the \textbf{Untrusted Environment}---it merely relies on the concatenation of two public values, the master public key \mpk{} and the identifier \cli{} of node \cvi{}.
Redeeming coins requires, instead, the \textbf{Hot Environment}.
The certificate \hatsigmai{} that associates \cvi{} to its public key \pkXi{} is a public parameter of the wallet, but we require the node's corresponding private signing key \skXi{} to sign a new transaction.
Compromising the secrets of node \cvi{} leads to compromising all its descendants, but none of the other nodes.
Finally, the master secret key \msk{} and the related signing key $\skX_0$ must be safely kept in \textbf{Cold Storage}.
We leverage these keys in the setup phase to prepare the authorization certificate \hatsigmai{} of \cvi{} and, for this reason, it is critical to the security of the wallet:
An attacker can use it to forge a certificate that associates any pair of keys to any target node and spend, in turn, the coins held in the entire wallet.

To summarize, Arcula defines a hierarchical deterministic wallet that benefits from the following properties:
\ifCONF{}
\begin{inlinelist}
\else
\begin{enumerate}
\fi
  \item Is secure against privilege escalation (\Cref{prp:secret_keys}).
  \item Generates every cryptographic key from an initial seed (\Cref{prp:determinism}).
  \item Enables identity-based public-key derivation so that users can dynamically derive new public keys without accessing their own private keys (see~\Cref{prp:public_keys}).
  \item Enables secret-key derivation so that users can sign transactions on behalf of their descendants.
  \item Does not rely on any particular digital signature scheme.
  \item The DHKA at the core of Arcula is secure under key indistinguishability and handles any directed acyclic graph encoding a partially ordered hierarchy.
    In addition, it allows dynamic modifications to the hierarchy (\emph{i.e.}, by adding or removing nodes, as we detail in~\Cref{sec:dynamic_changes}) and controlling the key assignment according to some temporal constraints (\Cref{sec:time_bound_deterministic_hierarchical_key_assignment}).
\ifCONF{}
\end{inlinelist}
\else
\end{enumerate}
\fi
\section{Arcula in the real world}\label{sec:arcula_in_the_real_world}
With Arcula, we design a future-proof HDW that brings identity-based signatures to the blockchain item and that, at the same time, is also suitable to the most widely used crypto-systems of today.
We aim to join theory and practice, to create a wallet that fulfills our current and future needs in the crypto-coins space.

For this reason, we constrain our design with as few cryptographic assumptions as possible.
Arcula works with any existentially unforgeable signature scheme and only requires the verification of a signature on an arbitrary message (\emph{i.e.}, the certificate that associates the signing key to their corresponding user).
This design makes it immediately compatible with the Ethereum blockchain, which implements a Turing-complete language, and with all the forks based on Bitcoin that allow the signature verification of arbitrary messages (\emph{e.g.}, Bitcoin Cash).
The original Bitcoin implementation, instead, does not allow such operation (in fact, it goes as far as disabling the operations of string concatenation and integer multiplication that, initially, it allowed).

In this section, we show how Arcula performs in the real world, and we show how to spend and receive funds, out of the box, on the Bitcoin Cash blockchain.
Next, we discuss the modifications that would make it compatible with the original Bitcoin protocol and how it is possible to disable public derivation to obtain unlinkability of transactions.

\subsection{Technical Implementation}\label{sub:implementation}
Our open-source implementation of Arcula is available online.\footnote{
\ifFULL{}
  Available at \url{https://github.com/aldur/Arcula}.
\else
  Available at \url{https://doi.org/10.5281/zenodo.2791637}.
\fi
}
We instantiate the underlying DHKA leveraged by Arcula with the pseudorandom function $\fash_{k}(x) = \hash(k \concat x)$ (where $\hash{}(x)$ is the hash function $\text{SHA3-256}(x)$) and the authenticated AES256 with Galois/Counter Mode (GCM) as the symmetric encryption scheme.
We generate a hierarchal deterministic wallet based on the tree defined in BIP43 and BIP44~\cite{marek2014bip43,marek2014bip44}, where the keys to different crypto-coins correspond to different subtrees, and each branch of the subtrees is a chain associated to a single account that contains multiple receiving addresses.
We obtain an initial seed \cSeed{} of $512$ bits by following the specification of BIP39~\cite{marer2013bip39} that generates a seed from a random mnemonic sequence.
We generate the wallet that we use in our tests by fixing the randomness of the mnemonic generation process to the result of the operation $\hash(\text{\texttt{correct horse battery staple}})$.

\subsection{Arcula in Bitcoin Cash}\label{sub:bitcoin_and_bitcoin_cash}

A Bitcoin transaction is a cryptographically signed statement that transfers some coins from a sender to a receiver.
The sender of the coins signs the transaction through her secret key to spend, in turn, the coins destined to the corresponding public key.
Every transaction specifies a locking and an unlocking script.
These scripts respectively state the necessary conditions to spend, in a future transaction, the coins being transferred (\emph{i.e.}, their locking condition) and provide the information required to redeem them (\emph{i.e.}, to unlock them as a result of a past transaction).
Both scripts are written through a stack-based language that allows simple mathematical operations, stack manipulations, and enables simple cryptographic primitives (\emph{i.e.} computing the result of a hash function and verifying a signature).

A typical Bitcoin locking script specifies the address of the receiver (usually through the hash of its public key) and requires him to provide a valid signature to redeem the coins being transferred.
More in detail, the locking and unlocking scripts of a standard Bitcoin transaction are defined as follows.
Uppercase monospace words indicate operations of the Bitcoin scripting language, while angular brackets enclose variable inputs.
\begin{description}
  \item[Locking: ] \OPDUP{} \OPHASH{} \texttt{<$\hash{(\pk)}$>} \OPEQUALVERIFY{} \OPCHECKSIG{}
  \item[Unlocking: ] \texttt{<$\sigma$>} \texttt{<\pk{}>}
\end{description}
Together, these scripts ensure that the public key \pk{} provided in the unlocking script is the pre-image of the hash $\hash(\pk)$ (the Bitcoin address) contained in the locking script;
then, verify the validity of the transaction signature $\sigma$ under the public key \pk{}.

In Arcula, instead, we identify the nodes of our wallet \cvi{} according to the master public key $\mpk{} = \pkX_0$ and to their public label \cli{}.
For this reason, an Arcula address is simply the concatenation of the byte representations of these values that we encode in the locking script.
The unlocking script, on the other hand, contains the certificate $\hatsigmai \sample \Sign_{\Sig}(\skX_0, (\pkX_i, \cli{}))$ and associating the signing public key \pkXi{} to the node \cvi{} with label \cli{}, and a signature $\sigma$ of the transaction verifiable through the public signing key \pkXi{}.
With Arcula, the locking and the unlocking scripts respectively become:
\begin{description}
  \ifFULL{}
  \item[Locking: ] \OPDUP{} \OPTOALTSTACK{} \texttt{<\cli{}>} \OPCAT{} \texttt{<\mpk{}>} \OPCHECKDATASIGVERIFY{} \\
  \OPFROMALTSTACK{} \OPCHECKSIG{}
\else
  \item[Locking: ] \OPDUP{} \OPTOALTSTACK{} \texttt{<\cli{}>} \OPCAT{} \texttt{<\mpk{}>} \OPCHECKDATASIGVERIFY{} \OPFROMALTSTACK{} \OPCHECKSIG{}
\fi
  \item[Unlocking: ] \texttt{<$\sigma$>} \texttt{<$\hatsigmai$>} \texttt{<\pkXi{}>}
\end{description}
The two scripts:
\begin{inlinelist}
\item Verify that the certificate $\hatsigmai$ is a valid signature of the message $(\pkXi, \cli)$ under the master public key \mpk{};
\item verify the validity of the transaction signature $\sigma$ under the signing public key \pkXi{}.
\end{inlinelist}
In particular, the locking script checks the validity of the certificate through the operation \OPCHECKDATASIGVERIFY{}, which allows the stack-based scripting language to validate a signature of an arbitrary message (the concatenation of $\mpk$ and \cli{} obtained through the operation \OPCAT{}).
The scripting language of the original Bitcoin does not implement such an operation yet.
Nonetheless, a significant portion of the Bitcoin community believes that its adoption would provide substantial benefits to the entire system, \emph{e.g.}, by enabling third-parties to store and verify independent messages on the blockchain.
For this reason, many Bitcoin forks (Bitcoin Cash, Bitcoin Ultimate, and Blockstream, to name a few), that aim at modernizing the protocol and at improving the stack-based language used in scripts, now implement this operation.

In our experiments, we focus, as an example, on Bitcoin Cash---the sixth cryptocurrency by market capitalization at the time of writing---and we evaluate Arcula on its test blockchain.
We first create a transaction\footnote{The transcripts of the transactions are available, respectively, at \url{https://bit.ly/2UI62tt} and \url{https://bit.ly/2UoQNGI}.} that locks $0.5$ BCH (the Bitcoin Cash crypto-coin) to a node of our wallet of~\Cref{sub:implementation}, identified through the master public key \mpk{} (also in the locking script) and the integer label $3$.
Next, we redeem the coins through a second transaction that provides the transaction signature $\sigma$ computed using the signing key \skXi{}, an appropriate certificate \hatsigmai{} signed by the master secret key, and the public signing key \pkXi{}.
We create both the signature and the certificate through the ECDSA signatures scheme on the \texttt{secp256k1} elliptic curve used in Bitcoin, and we encode the integer label \cli{} of the node \cvi{} with $4$ bytes.

\begin{table}[tpb]
  \centering
  \small
  \caption{The script bytes sizes of a transaction to a standard Bitcoin address and to an Arcula address.}\label{tab:script_sizes}
  \begin{tabular}{c c c c}
    \toprule
    Address type & Locking Script & Unlocking Script & Total \\
    \midrule
    Standard   & $24$      & $106$      & $130$ \\
    Arcula    & $43$      & $179$      & $222$ \\
    \bottomrule
  \end{tabular}
\end{table}

\subsection{Transaction Costs}\label{sub:transaction_costs}
To study the costs of Bitcoin transactions to an Arcula address, we analyze the amount of storage that they require on the blockchain.
Every Bitcoin transaction devolves a small amount of fees to the system to incentive its inclusion in the next block of the chain.
Fees are usually measured in coins per byte, and, for this reason, the size of a transaction on the Bitcoin wire protocol is directly related to the amount of fees that it should pay to be included in the blockchain.
In particular, the length of the locking and unlocking scripts influences directly the final transaction cost.
\Cref{tab:script_sizes} compares the sizes, in bytes, of the locking and unlocking scripts of standard Bitcoin transactions and to an address of our wallet.
Every operation of the stack-based scripting language is encoded with a single byte;
a standard Bitcoin address is the result of a hash function that outputs $20$ bytes;
the ECDSA signature and the public key in the unlocking script require, respectively, $73$ and $33$ bytes.
By summing these values up, we find that the locking script of a transaction to a standard Bitcoin address is $24$ bytes long ($4$ script operations plus the receiver address) while the unlocking scripts take $106$ bytes (the ECDSA signature and its associated public key).
In Arcula, on the other hand, the locking script encodes $6$ operations, the identifier of a node (that we encode with $4$ bytes), and the cold storage public key ($33$ bytes, as opposed to its $20$ bytes hash), for a total of $43$ bytes.
The unlocking script, instead, contains two ECDSA signatures (one for the transaction and one for the certificate) and the signing public key;
as a result, it is $179$ bytes long.
Overall, the size of the locking and unlocking scripts for a transaction to an Arcula address is $222$ bytes, $70\%$ longer than the standard address counterparts.

In particular, the Bitcoin users aim at minimizing the size of the locking script, as its associated fees will be paid by the sender of the transaction, \emph{e.g.}, the customer of an online service, and the service providers usually aim at minimizing these costs.
Bitcoin solves this issue through the pay to script hash mechanism, proposed in BIP16~\cite{andresen2012bip16}, that reduces the size of any locking script to a constant at the cost of longer unlocking scripts.
The intuition is that instead of specifying the full locking script, the users can constrain the coins of a transaction by locking them to the hash of the original script;
then, in the unlocking script, they can provide both the pre-image of the hash, \emph{i.e.}, the full locking script, and its required inputs.
This approach brings several advantages.
First, any locking script can be expressed with a constant byte size that results in a fixed cost for the sender.
Second, it hides the details of the locking script until the users reveal the pre-image of the hash in an unlocking script, \emph{i.e.} when they redeem the coins sent by the transaction.
Finally, the Bitcoin protocol proposes a way to encode the pay to script hash locking scripts into standard Bitcoin addresses, so that exchanging transactions of this kind is entirely transparent to the software used by the sender.
By using the pay to script hash mechanism, any user can send a transaction to an Arcula address through her favorite Bitcoin wallet, in a transparent way that does not require any specific software modification to it.
More in details, an Arcula pay to script hash transaction is defined as follows, where the script that we input to the hash function is the locking script of a transaction to an Arcula address that we have seen before:
\begin{description}
\ifFULL{}
  \item[Script: ] \OPDUP{} \OPTOALTSTACK{} \texttt{<\cli{}>} \OPCAT{} \texttt{<\mpk{}>} \OPCHECKDATASIGVERIFY{} \\
  \OPFROMALTSTACK{} \OPCHECKSIG{}
\else
  \item[Script: ] \OPDUP{} \OPTOALTSTACK{} \texttt{<\cli{}>} \OPCAT{} \texttt{<\mpk{}>} \OPCHECKDATASIGVERIFY{} \OPFROMALTSTACK{} \OPCHECKSIG{}
\fi
  \item[Locking: ] \OPHASH{} \texttt{<\hash{(Script)}>} \OPEQUAL{}
  \item[Unlocking: ] \texttt{<$\sigma$>} \texttt{<$\hatsigmai$>} \texttt{<\pkXi{}>} \texttt{<Script>}
\end{description}

The pay to script hash mechanism reduces to $22$ bytes ($2$ operations and a $20$ bytes hash) the size of the locking script and, equivalently, the amount of fees that users have to spend to send funds to an Arcula address.
The size of the unlocking script, on the other hand, affects the fees that the users of Arcula need to pay when spending their coins.
In particular, when using pay to script hash, this amount of fees is slightly larger than the one required for a traditional Bitcoin transaction.
In many cases, however, the benefits that arise with Arcula justify the increase in the transaction cost.
An e-commerce marketplace, as an example, can leverage Arcula's public key derivation to dynamically derive new addresses (\emph{e.g.}, one for each product of her catalog) in an entirely untrusted environment (\emph{e.g.}, an online web-server) while keeping every signing keys at rest in trusted storage.
As a result, the provider obtains the flexibility of handling incoming payments on dynamic addresses and minimizes the risk of losing the coins associated with them.
When compared with the financial costs associated with this risk, the additional fees required by the Arcula transactions are negligible.
The public key derivation also brings other significant benefits.
Many financial regulations require, indeed, companies to be accountable for all the payments that they receive.
With Arcula, an auditor can reach this goal by merely inspecting the blockchain while looking for any address that contains the master public key \mpk{} that identifies the company.
Finally, many enterprises leverage $m$-of-$n$ signatures, where redeeming a transaction requires $m$ valid signatures among $n$ authorized public keys.
Their goal is to enforce the internal structure of the company (\emph{e.g.}, so that either managers or employees can sign transactions) or to divide the responsibility of spending coins evenly.
The unlocking scripts of $m$-of-$n$ transactions have considerable size:
They contain $m$ signatures and $n$ public keys.
By leveraging Arcula and enforcing an appropriate hierarchy that reflects their internal structure, these companies could reduce the size of the unlocking scripts to only two signatures (the transaction signature and the certificate) and two public keys (the master and signing public keys).

\subsection{Optimizations and compatibility with Bitcoin}\label{sub:optimizations}
The current implementation of Arcula does not require any modification to the underlying protocols and blockchains.
Nevertheless, we also propose a set of optimizations that, through minimal modifications to these protocols, reduce both the cost of transactions to Arcula addresses and the amount of storage required on the blockchain.
We begin by noting that any authorization certificate \hatsigmai{} can be used more than once.
 % by the corresponding signing key \pkXi{}.
For this reason, the first optimization that we propose is to \emph{cache} the certificate \hatsigmai{} as soon as it appears for the first time in an unlocking script.
Then, any subsequent transaction signed by \skXi{} could specify a pointer to the certificate (\emph{e.g.}, with a shorter hash) instead of the certificate itself and, in turn, reduce the size of the unlocking script.
As an example, by pointing to the certificate with a $20$ bytes hash, we would reduce the size of the Arcula locking and unlocking scripts to be roughly $20$ bytes longer than their traditional counterparts.
Implementing this optimization requires a new operation in the scripting language to retrieve the certificate from the cache and to verify its validity.
On the other hand, if we allow for more complex modifications, we can change the signature scheme of the underlying protocols to reduce these space requirements to their optimal value further---a single signature per transaction.
Arcula can be implemented with a single signature by leveraging a sanitizable signature scheme~\cite{ateniese2005sanitizable}, \emph{i.e.} a scheme where an authorized party can modify a fraction of the message signed without interacting with the original signer.
The intuition is to combine the certificates with the signatures that authorize transactions:
Now, the certificate of user \cvi{} that associates her to the signing public key \pkXi{} also includes an additional modifiable portion that will be filled with the transaction details.
To spend their coins, the users leverage their sanitizable key to replace the blank transaction with the details that they intend to sign.\footnote{The sanitizable keys can be hierarchically deployed by leveraging a second instance of DHKA.}
In their work, Ateniese \emph{et al.}~\cite{ateniese2005sanitizable} show how to construct a sanitizable signature scheme by combining any signature scheme with a chameleon hash function.
This construction would allow Arcula to be used with the traditional Bitcoin blockchain by implementing the sanitizable signatures on top of the ECDSA signature scheme that it already uses.
In addition, it would not change the expressiveness of the Bitcoin scripting language:
Instead of enabling the verification of signatures on arbitrary messages, it would simply extend the signature verification protocol to account for the certificate embedded in the sanitized signatures.

\subsection{Unlinkability of Transactions}\label{sub:transactions_unlinkability}
Individual users of hierarchal deterministic wallets are typically not interested in public key derivation.
Differently from enterprises and e-commerce marketplaces, for instance, they simply rely on HDW to recover their keys in case of hardware failure or catastrophic loss.
On the other hand, they are often interested in achieving unlinkability of their transactions, \emph{i.e.} in making sure that multiple transactions sent to their wallet can not be correlated together by an observer that passively monitors the blockchain.
In other words, they typically desire to trade the derivation of public keys in an untrusted setting for the ability to receive payments on uncorrelated pseudonyms.

Arcula allows them to reach this goal.
In more detail, these users can ignore the identity-based public key derivation that Arcula provides (and its associated master public key \mpk) and identify the nodes of the wallet with their public signing key \pkXi{}.
On the blockchain, they can receive standard transactions (costing standard transaction fees) on the public signing key \pkXi{} and then sign new transactions to redeem the coins through the corresponding private key \skXi{}.
We generate this pair of keys through a DHKA scheme that is secure under key indistinguishability.
This means that every public signing key in the wallet is unlinkable from the others, because the DHKA cryptographic keys \cxi{} that we use as randomness to generate them are, in turn, indistinguishable from random strings of the appropriate length.
As a result, Arcula provides a provably secure alternative to the hardened mode of BIP32:
Individual users can generate as many pseudonyms as they need by branching or deepening the DAG that encodes their hierarchy, and then leverage the DHKA to generate keys and reliably recover them in case of loss.
Note that this modified version of Arcula does not require validation of signatures of arbitrary messages, enabling its usage with any blockchain system, including Bitcoin.

In addition, Arcula also allows users to achieve unlinkability of transactions while maintaining identity-based public key derivation.
The intuition is to use a chain code $c$, private to the environment where we execute the public derivation, to \emph{perturb} the master secret and public keys so that they look uncorrelated from the original keys to a passive observer.
We perform the perturbation once for each node in the wallet.
As a result, we associate them with a set of perturbed pairs of keys, labeled $(\mski, \mpki)$, that we use to sign the certificate $\hatsigmai$ that associates their public signing key \pkXi{} to their identity \cvi{}.
Now, we can address Bitcoin payments for the node \cvi{} to the $i$-th perturbation \mpki{} of the master public key \mpk{}, that is uncorrelated from any other key of the wallet.

More in detail, let \cGroupg{} be a generator of the elliptic curve used in Bitcoin's ECDSA signature scheme.
The master public key is defined as $\mpk = \mpk_{0} = \pkX_0{} = \cGroupPub{\skX_0{}}$.
Let $c$ be the secret chain code and let \fash{} be a pseudorandom function.
We create the $i$-th perturbed key \mpki{} of the master public key \mpk{} as follows:
\begin{equation*}
  \mpki{} = \cGroupPub{\skX_0{} + \fash_{c}{(\cli{})}} = \cGroupPub{\skX_0{}} \cdot \cGroupPub{\fash_{c}{(\cli{})}} = \mpk \cdot \cGroupPub{\fash_{c}{(\cli)}},
\end{equation*}
where \cli{} is the label of node \cvi{}.
As long as the chain code $c$ is private, this construction ensures the unlinkability of transactions sent to the perturbed addresses~\cite{maxwell2011deterministic}.

In Arcula, we modify the~\Cref{enum:sign_oracle_proof_inner} of \Cref{const:cert-arcula} to sign the certificates $\hatsigmai \sample \Sign_{\Sig}(\mski, \pkXi)$ with the perturbed secret key $\mski{} = \skX_0{} + \fash_{c}{(\cli{})}$.
Note that we remove the label \cli{} from the certificate since now every pair of perturbed keys is uniquely associated with precisely one pair of signing keys, and the perturbation already takes into explicit consider the label \cli{} of node \cvi{}.
Finally, we replace the master public key $\mpk$ in the locking script with the $i$-th perturbed key \mpki{}, that verifies the certificate \hatsigmai{}, as follows:
\begin{description}
  \item[Locking: ] \OPDUP{} \OPTOALTSTACK{} \texttt{<\mpki{}>} \OPCHECKDATASIGVERIFY{} \ifFULL{} \\ \fi{} \OPFROMALTSTACK{} \OPCHECKSIG{}
  \item[Unlocking: ] \texttt{<$\sigma$>} \texttt{<$\hatsigmai$>} \texttt{<\pkXi{}>}
\end{description}
As a result, all the Arcula addresses of the same wallet look uncorrelated when they appear in the locking script of a transaction.
The perturbed private keys \mski{} are effectively equivalent to the original master secret key:
An attacker that compromises any perturbed key can recover the master secret key and compromise the entire wallet by forging new certificates for key pairs that he controls.
For this reason, the perturbed keys shall be kept in cold storage, or better yet, destroyed after the generation of the corresponding certificate.

To conclude, we briefly discuss how to derive, deterministically, the chain code $c$.
We propose to assign a different chain code $c_i$ to each node \cvi{} of the wallet by running our DHKA a second time:
In this way, an attacker that compromises a node in the hierarchy can only uncover the public identifiers of the nodes in its subtree, but would not gain any knowledge about the others in the hierarchy.
\section{Conclusions}\label{sec:conclusions}
In this work, we presented Arcula, a new hierarchical deterministic wallet (HDW) that brings identity-based signatures to the blockchain, and that is secure against privilege escalation.
We first developed a key indistinguishable deterministic hierarchical key assignment (DHKA) scheme, that we use to deterministically generate the set of cryptographic keys at the core of our wallet.
As a result, an attacker that compromises an arbitrary number of users in the hierarchy can not escalate his privileges and compromise other users higher in the hierarchy.
In addition, our wallet allows us to dynamically derive new addresses for receiving payments in an entirely untrusted environment, to recover every cryptographic key from an initial seed provided by the user, and also to spend coins on behalf of users lower in the hierarchy.
Our design of Arcula considers the legacy and future requirements of modern blockchains.
In particular, Arcula is independent of the underlying signature scheme, and it works on top of any protocol that allows the verification of signatures on an arbitrary message (\emph{e.g.}, Bitcoin Cash or Ethereum).
For these reasons, we hope that the outcomes of this work will be twofold:
To provide the secure and efficient hierarchical deterministic wallet that we need today and to propose a future-proof design that supports the financial applications and tools of enterprises and companies at scale.

\ifCONF{}
    \bibliographystyle{IEEEtran}
    \bibliography{IEEEabrv,main}
\else
    \bibliographystyle{plainnat}
    \bibliography{main}
\fi

\ifCONF{}
\appendices{}
\else
\appendix
\fi
\section{Further Preliminaries}\label{sec:further_preliminaries}

\subsection{Pseudorandom Function (PRF) Family}\label{sub:pseudorandom_function_prf_family}
Let $\{\cKK_\secpar, \cXX_\secpar, \cY_\secpar\}_{\secpar \in \NN}$ be a sequence of sets.
For $\secpar \in \NN$, a PRF family \prff{} is a set of functions such that $\fash{}_k : \cXX_\secpar \to \cY_\secpar$ and each function is evaluable by a deterministic polynomial time algorithm $\fash$, \emph{i.e.}, $\fash(k, \cdot) = \fash_k(\cdot)$.

Let \Rand{} be the set of all functions from $\cXX_\secpar$ to $\cY_\secpar$.
For security, we require that a function randomly sampled from \prff is indistinguishable by a function randomly sampled from \Rand{}.
\begin{definition}[Pseudorandomness]\label{def:prf}
A PRF family \prff{} is pseudorandom if for every PPT adversary \Adv{} we have:
\begin{equation*}
    \Biggl| \prob{\GamePRFZero = 0} - \prob{\GamePRFOne = 1} \Biggr| \leq \negl,
\end{equation*}
where the two experiments \GamePRFZero{} and \GamePRFOne{} are defined in the following way:
\begin{pchstack}[center]
    \procedure{\GamePRFZero}{%
        k \sample \cKK_\secpar \\
        d \sample \Adv^{\fash_k(\cdot)}(\secparam) \\
        \pcreturn d
    }
    \pchspace
    \procedure{\GamePRFOne}{%
        f \sample \Rand{} \\
        d \sample \Adv^{f(\cdot)}(\secparam) \\
        \pcreturn d
    }
\end{pchstack}
\end{definition}

\noindent In this paper, we are interested in PRF families such that $\cKK{}_\secpar = \cY_\secpar = \bin^\secpar$.

\subsection{Symmetric Encryption Scheme}\label{subsec:sim-enc}
We follow the definition of symmetric encryption scheme provided by Atallah \emph{et al.}~\cite{atallah:2009:dek:1455526.1455531}.
A symmetric-key encryption scheme $\Pi = (\Gen{}, \enc, \dec{})$ with message space \cMM{} is a triple of polynomial-time algorithm defined in the following way:
    \begin{description}
        \item[\sGen{}:] The randomized key generation algorithm takes as input a security parameter \secparam{} and outputs a secret key \sk{}.
        \item[\sEnc{}:] The deterministic (possibly randomized) encryption algorithm takes as input a secret key \sk{}, a message $m \in \cMM{}$, and outputs a ciphertext $c$.
        \item[\sDec{}:] The deterministic decryption algorithm takes as input a secret key \sk{}, a ciphertext $c$, and outputs a message $m$.
    \end{description}

For correctness, we require that honestly generated ciphertexts must decrypt correctly.
\begin{definition}[Correctness of symmetric encryption]\label{def:enc-correctness}
  A symmetric encryption scheme $\Pi = (\Gen{}, \enc, \dec{})$ with message space \cMM{} is correct if $\forall{\secpar} \in \NN, \forall{m} \in \cMM$:
\begin{equation*}
    \condprob{\dec(\sk, \enc(\sk,m)) = m}{\sk \sample \sGen{}} = 1
\end{equation*}
\end{definition}

For security, we are interested in semantic security: It must be infeasible to distinguish between an encryption of a message $m$ from one of a random message.

\begin{definition}[Semantic Security]\label{def:enc-def}
    A symmetric encryption scheme $\Pi = (\Gen, \enc, \dec)$ with message space $\cMM$ is semantically secure if for every PPT adversary \Adv{} we have:
    \begin{equation*}
        \Biggl| \prob{\GameEncPi = 1} - \frac{1}{2} \Biggr| \leq \negl,
    \end{equation*}
    where $\GameEncPi$ is defined in the following way:
\begin{description}
    \item[Setup:] The challenger runs $\sk{} \sample \sGen{}$.
    \item[Challenge:] The adversary specifies a message $m_0 \in \cMM$.
        The challenger picks a random bit $b^* \in \bin$.
        If $b^* = 0$, then it computes $c^* \sample \enc(\sk, m_0)$;
        otherwise, it sets $c^* \sample \enc_\sk(m_1)$, where $m_1 \sample \cMM$.
        The challenger returns $c^*$ to \Adv{}.
    \item[Guess:] The adversary outputs a bit $b \in \bin$. If $b = b^*$ return $1$; otherwise return $0$.
\end{description}
\end{definition}

\section{Security Model of Deterministic Hierarchical Key Assignment}\label{sec:deterministic_wallets_through_dynamic_and_efficient_key_management}
The correctness of a DHKA scheme requires that any user \cvi{} should be able to derive, correctly, the secret key \cxj{} of any user $\cvj \in \desc(\cvj)$ lower in the hierarchy.

\begin{definition}[Correctness of DHKA]\label{def:kidhka}
  A DHKA $\Pi=(\Set,\Derive)$ with seed space $\cSS$ is correct if for every DAG $\cG = (\cV, \cE)$, $\forall \secpar \in \NN$, $\forall v_i \in \cV$, $\forall \cvj \in \desc{}(\cvi)$, $\forall \sSeed{}$:
  \begin{equation*}
    \prob{\cx{j} = \sDerive{}} = 1,
  \end{equation*}
  where $\PubSec = \sdSet$, $(\cSi, \cxi) = \Sec(\cvi)$, and $(\cS{j}, \cx{j}) = \Sec(\cvj)$.
\end{definition}

We now formalize the security level of the scheme.
We adapt the security definition originally defined by Atallah \emph{et al.}~\cite{atallah:2009:dek:1455526.1455531} to account for the determinism in our scheme.
We define the set of ancestors \danc{} of a node \cvi{} to be the set of nodes \cvj{} such that there exists a path $w$ from \cvj{} to \cvi{} in \cG{}.

\begin{definition}[Key Indistinguishability of DHKA]\label{def:sec_dka}
  A DHKA $\Pi = (\Set, \Derive{})$ with seed space $\cSS$ is key indistinguishable if for every PPT adversary \Adv{} and every DAG $\cG = (\cV, \cE)$:
  \begin{equation*}
    \biggl| \prob{\GameSKeyInd{} = 1} - \frac{1}{2} \biggr| \leq \negl{},
  \end{equation*}
  where \GameSKeyInd{} is defined in the following way:
  \begin{description}
    \item [Setup:] The challenger receives a challenge node $\cvstar{} \in \cV$ from the adversary \Adv{}.
      The challenger samples $\cSeed \sample \cSS$, then runs \sdSet{}, and gives the resulting public information \Pub{} to the adversary \Adv{}.
      The challenger samples a random bit $b^{*} \sample \bin$:
      If $b^{*} = 0$, it returns to \Adv{} the cryptographic key \cxstar{} associated to node \cvstar{};
      otherwise, it returns a random key \cxstarbar{} of the corresponding length.
    \item[Query:] The adversary has access to a corrupt oracle $\CorrOracle(\cdot)$. On input $\cvi{} \notin \anc(\cvstar{})$, the challenger retrieves $\csxi = \aSec$ and sends \cSi{} to \Adv{}.
    \item[Guess:] The adversary outputs a bit $b \in \bin$. If $b = b^*$ return $1$; otherwise return $0$.
  \end{description}
\end{definition}

\begin{remark}
We note that the adversary \Adv{} depicted in \GameSKeyInd{} is a static adversary who chooses the challenge node \cvstar{} before the experiment begins.
Ateniese \emph{et al.}~\cite[Theorem 1]{ateniese2012}, however, prove that any hierarchical key assignment scheme secure (in the sense of \GameSKeyInd{}) against a static attacker is also secure against an adaptive attacker, \emph{i.e.}, against an adversary that adaptively chooses the challenge node \cvstar{}.
The authors prove that the two security models are polynomially equivalent since there exists a reduction between the static and the adaptive adversaries.
The static adversary can simply guess the challenge node \cvstar{} of the adaptive adversary and abort the simulation if the guess is incorrect.
For these reasons, we discuss the security of any DHKA scheme only in the setting of a static attacker.
\end{remark}

\subsection{The DHKA scheme}\label{sub:dhka_implementation}
This section describes the implementation of our deterministic hierarchical key assignment scheme over any DAG \cG{} encoding an access hierarchy.

We assume, without loss of generality, that:
\begin{inlinelist}
\item There exists a unique \emph{root} node $\cvzero{} \in \cV{}$ of \cG{}, \emph{i.e.} the most-privileged node of the hierarchy encoded by \cG{} that can derive the keys of any other node.
  For any DAG \cG{}, it is always possible to elect a root node \cvzero{}.
  Since \cG{} is a DAG, \cvzero{} shall be one of the minimal nodes in a topological ordering of \cG{} and, equivalently, \cvzero{} shall have no ancestors.
  If two or more nodes \cvj{} have no ancestors, then it is always possible to construct a new graph $\cG{}' = (\cV \cup \set{\cvzero{}}, \cE \cup \set{(\cvzero{}, \cvj) \mid \cvj{} \text{ has no ancestors}})$ such that the access hierarchy encoded by $\cG{}'$ is equivalent to the one of \cG{}, where the new node \cvzero{} in $\cG'{}$ is the root of the graph (and has no associated users).
\item That every node \cvj{} has a fixed \emph{parent} node in the hierarchy, \emph{i.e.} a node \cvi{} such that the edge $\sE{}$.
  As an example, we fix the parent node \cvj{} of \cvi{} to be the first ancestor of \cvi{} in any ordering of the nodes of the graph \cG{} (\emph{e.g.}, obtained with a depth-first-search of the graph) such that \sE{}.
\end{inlinelist}

At a high level, we build on the randomized hierarchical key assignment scheme of Atallah \emph{et al.}~\cite{atallah:2009:dek:1455526.1455531} where each node \cvi{} of the hierarchy is identified by a random label \cli{} and holds a random secret information \cSi{}, that it will use to generate its own cryptographic key and to derive the keys of the nodes lower in the hierarchy.
In our scheme, we modify the original design so that both the label \cli{} and the secret information \cSi{} are deterministic.
We label each node through its index\footnote{
We will extend the node labels with a version number when handling dynamic changes to the hierarchy of the DHKA.
We refer the reader to~\Cref{sec:dynamic_changes} for more details.
} (\emph{i.e.}, $\cli{} = \cvi{}$) and we derive its secret information \cSi{} deterministically (through a pseudorandom function) from the secret information of its parent.
We formally define our implementation of DHKA as follows.
\begin{construction}\label{const:dka}
  Let \prff{} and $\Epsilon{} = (\Gen, \enc, \dec)$ be respectively a family of pseudorandom functions and a symmetric key encryption scheme.
  Let \sG{} be a directed acyclic graph representing an access hierarchy.
  We build a DHKA scheme in the following way:
  \begin{description}
    \item[$\Set{}(1^\secpar, \cG{}, S_{-1})$:] On input the security parameter, a directed acyclic graph \sG{}, and an initial seed $S_{-1}$, the algorithm proceeds as follows:
      \begin{enumerate}
        \item Compute \scSzero{} for $\cv{0} \in \cV$, where \sclzero{} and \cvzero{} is the \emph{root} of the directed acyclic graph \cG.\label{item:difference_root}
        \item For each vertex $\cvi \in \cV$ and $\cvj \in \cV$ such that \cvj{} is the \emph{parent} of \cvi{}, compute \scSi{} where \sli{}.\label{item:difference_node}
        \item For each vertex $\cvi \in \cV$ compute \sti{} and \sxi{}.
        \item For each edge \sE, compute \srij{} and \syij{}.\footnote{We implicitly assume that the PRF output space and the symmetric encryption key space have the same distribution. In alternative, \crij{} can be used as randomness of key generation algorithm \Gen{}.}
      \end{enumerate}
      Finally, the algorithm returns the public mapping \sPub{} and the secret mapping \sSec{}, defined as:
      \begin{gather*}
        \dPub \\
        \dSec
      \end{gather*}
    \item[$\Derive{}(\cG{}, \Pub{}, v_i, v_j, S_i)$:] On input a directed acyclic graph \sG{}, a public mapping \Pub{}, two nodes $v_i, v_j \in \cV$, and a seed $S_i$, the algorithm proceeds as follows:
      \begin{enumerate}
        \item If there is no path from \cvi{} to \cvj{} in \cG{}, return $\bot$;
        \item If $i = j$, retrieve \cli{} from \Pub{} and return \sxj{};
        \item Otherwise, compute \sti{} and set $\bar{i} = i$ and $t_{\bar{i}} = t_i$; then
          \begin{enumerate}
            \item Let $\bar{j}$ be the successor of $\bar{i}$ $\text{in}$ the path from \cvi{} to \cvj{}. \label{itm:derive}
            \item Retrieve $\cl{\bar{j}}$ and $\cy{\bar{i}\bar{j}}$ from \Pub{}
            \item Compute $\crr{\bar{i}\bar{j}} = \fash_{\ct{\bar{i}}}{(10 \concat \cl{\bar{j}})}$ and $\ct{\bar{j}} \concat \cx{\bar{j}} = \dec_{\crr{\bar{i}\bar{j}}}{(\cy{\bar{i}\bar{j}})}$.
            \item Set $\bar{i} = {\bar{j}}$ and $\ct{\bar{i}} = \ct{\bar{j}}$.
            \item If $\bar{j} = j$ then return $\cxj$; otherwise repeat from \Cref{itm:derive}.
          \end{enumerate}
      \end{enumerate}
  \end{description}
\end{construction}

The proposed key assignment scheme is entirely deterministic.
In particular, it differs from the design of Atallah \emph{et al.} at the~\Cref{item:difference_root,item:difference_node} of the \Set{} algorithm in~\Cref{const:dka}.
The original key assignment scheme draws the values \cli{} and \cSi{} (respectively, \clzero{} and \cSzero) at random.
In our case, instead, we deterministically derive them from the identifier \cvi{} of the node and from the secret information \cSj{} of its \emph{parent} \cvj{} (respectively, from the seed \cSeed{}).

\ifFULL
\paragraph{Computation and space complexity}
\fi
The efficiency of the
\ifCONF HKA \fi
scheme
\ifCONF of Atalla \emph{et al.}~\cite{atallah:2009:dek:1455526.1455531} \fi
is linear in time and space, respectively, to the key derivation distance and the size of the graph.
Let $w$ be the shortest path between \cvi{} and $\cvj{} \in \desc{(\cvi{})}$:
Deriving \cxj{} by starting from \cSi{} requires $|w|$ invocations \fash{} and $|w|$ invocations of $\dec$.
For space complexity, each node \cvi{} in \cV{} is required to store a single secret \cSi{}---the private storage required by each node is proportional to the size \secpar{} of the security parameter.
On the other hand, the public information holds the mapping between nodes and labels and the encrypted information associated with each edge.
As such, the overall space required is linear to $\secpar |\cV| + \secpar |\cE|$.
That said, we note that in our case
\ifCONF (\Cref{const:dka}) \fi
the mapping between nodes and labels is the identity function and that we can further reduce the storage requirements by leveraging the deterministic derivation:
Any parent node \cvj{} can directly derive the secret information \cSi{} of its descendant \cvi{} and, for this reason, we can avoid storing any encrypted information on the edge that connects them.
As a result, we can reduce the size of the encrypted information on the edges and only store them for any node \cvi{} such that there exists an edge \sE{} and \cvi{} is not the parent of \cvj{} and as such cannot deterministically derive the secret value \cSj{} by starting from its own secret value \cSi{}.
With this optimization in place, our scheme is comparable to a tree-based hierarchical key assignment scheme~\cite{crampton2017} where we store the additional derivation keys as encrypted information on the edges instead of storing them as secrets within each node that requires them.
Finally, if the key generation and derivation processes happen on the fly (\emph{i.e.}, when the entire process starts from the seed), then the only private storage required is proportional to the length of the initial seed \cSeed{}, \emph{i.e.}, to the length of the security parameter \secpar{}.

\begin{remark}
  At first glance, it might seem that fixing the randomness of the \Set{} algorithm of the HKA by Atallah \emph{et al.}~\cite{atallah:2009:dek:1455526.1455531} is sufficient to enforce its determinism.
  We remark here that such a solution, alone, does not guarantee this result.
  When we fix the randomness of the \Set{} algorithm of the HKA we are implicitly fixing an ordering on the sampling of the secret values \cSi{} of each node \cvi{}:
  Sampling at random \cSi{} before \cSj{}, as opposed to sampling \cSj{} before \cSi{}, will result in different secret values assigned to each node.
  For this reason, the HKA with fixed randomness would also require additional public information about the ordering of the nodes of the hierarchy.
  Our DHKA, instead, deterministically generates the secret values according to the structure of the hierarchy and not to any ordering of its nodes.
  This approach allows us to design a deterministic scheme that does not require any additional public information and that, furthermore, can take the determinism into account to reduce the amount of encrypted information stored on the edges of the hierarchy.
  % Besides, both the original HKA by Atallah \emph{et al.} and our DHKA handle time-bound constraint (\Cref{sec:time_bound_deterministic_hierarchical_key_assignment}) and dynamic changes to the access hierarchy (\emph{e.g.}, by adding or removing any node or edge, as we detail in~\Cref{sec:dynamic_changes}).
  % When we leverage the DHKA as a building block of Arcula, we are particularly interested in appending new leaves to the hierarchy (\emph{e.g.}, one new node for each incoming payment).
  % In our scheme, we make sure that each new node \cvi{} will always be deterministically assigned the same secret value \cSi{}, regardless of the other nodes of the graph and of the time of its addition to the hierarchy.
\end{remark}

We conclude this section by establishing the following result.
% , that we prove in~\Cref{subsec:proof_dhka}.
\begin{theorem}\label{th:d_key_ind}
  Let \prff{} and $\Epsilon{} = (\Gen{}, \enc, \dec{})$ be respectively a pseudorandom function family and a symmetric encryption scheme.
  If \prff{} is pseudorandom (\Cref{def:prf}) and \Epsilon{} is semantically secure (\Cref{def:enc-def}), then the DKHA scheme $\Pi$ from \Cref{const:dka} is key indistinguishable.
\end{theorem}
\begin{proof}
 We prove the theorem by contradiction, using a hybrid argument.
Let \cvstar{} be the challenge chosen by an adversary \Adv{} in the game \GameSKeyInd{}.
We define the following hybrid experiments:
\begin{description}
    \item[$\Hyb_{-1}$:] is exactly the game \GameSKeyInd{}.
    \item[\HybZero:] is the same as $\Hyb_{-1}$, except that the secret \cSzero{} of the root node $\cvzero{} \in \ancstar{}$  is sampled at random.
    \item[$\Hyb^{(a)}_i$:] is the same as $\Hyb^{(c)}_{i-1}$ (for $i = 1$ is the same as \HybZero{}), except that $\ct{i-1}$, $\cx{i-1}$ associated to the node $\cv{i-1}\in \ancstar{}$ and $\cS{i}$ of the node $\cv{i}\in \ancstar{}$ are sampled at random.
    \item[$\Hyb^{(b)}_i$:] is the same as $\Hyb^{(a)}_{i}$, except that \crij{} associated to the edge $(\cvi{}, \cvj{})$ (where $\cvi{}, \cvj{} \in \ancstar{}$) is sampled at random.
    \item[$\Hyb^{(c)}_i$:] is the same as $\Hyb^{(b)}_{i}$, except that \cyij{} associated to the edge $(\cvi{}, \cvj{})$ (where $\cvi{}, \cvj{} \in \ancstar{}$) is an encryption of a random message, \emph{i.e.}, $\cyij{} \sample \enc_{\crij}(\hat{m})$ where $\hat{m}$ is sampled at random.
\end{description}

Our DHKA is identical to the HKA of Atallah \emph{et al.}~\cite{atallah:2009:dek:1455526.1455531}, except that the secret $\cS{i}$ of a node \cvi{} is computed by evaluating $\cS{i} = \fash_{\cS{j}}(11\concat\cli{})$ where $\cS{j}$ is the secret of the parent \cvj{} of \cvi{} (in~\cite{atallah:2009:dek:1455526.1455531} each \cSi{} is sampled at random).
Hence, the proof is analogous to~\cite[Theorem 5.3]{atallah:2009:dek:1455526.1455531} except that we need to prove that each $\cS{i}$ is indistinguishable from random.
For this reason, we modify the game $\Hyb^{(a)}_{i}$ (defined in~\cite[Theorem 5.3]{atallah:2009:dek:1455526.1455531}) in such a way that the secret $\cS{i}$ is sampled at random too (in addition to $\ct{i-1}, \ck{i-1}$).
Then, we prove the same result for the root node \cvzero{} by adding an additional game $\Hyb_{-1}$ and by showing $\Hyb_{-1} \approx_c \HybZero$.

\begin{lemma}\label{lem:inductive_proof_dka}
    Let \prff{} be a secure pseudorandom function, then $\Hyb_{-1} \approx_c \HybZero$.
\end{lemma}
\begin{proof}
    We assume that there exists a DAG $\cG = (\cV, \cE)$ and a distinguisher \Dist{} that has a non-negligible advantage in distinguishing between $\Hyb_{-1}$ and $\HybZero$.
    Then, we build an adversary \Adv{} that distinguishes \GamePRFZero{} and \GamePRFOne{} as follows:
    \begin{enumerate}
      \item \Dist{} outputs the challenge $\cvstar{}$.
      \item \Adv{} simulates \Set{} as follows:
          For the root node \cvzero{}, set $\cSzero{} = \PRFOracle(11||\cl{0})$.
          For any other node $\cvj{}$, compute $\cS{j}$ as described in \Cref{const:dka}.
          Then, for each node $\cvi{} \in \cV$ and for each edge \sE{}, compute the secret values \sti{}, \sxi{}, \srij{}, \syij{} as described in \Cref{const:dka}.
          \Adv{} sets $\cxstar^0 = \cxstar{}$ and $\cxstar^1 = \cxstarbar{}$ where $\cxstarbar$ is sampled at random.
          Finally, \Adv{} sends \Pub{} and $\cxstar^d$ to \Dist{} where $d$ is a random bit.
      \item \Adv{} answers any $\CorrOracle^{\Pi}(\cvi{})$ query by returning \cSi{}.
      \item \Dist{} outputs a bit $d'$ and $\Adv{}$ completes the simulation of the experiments $\Hyb_{-1}$ and $\HybZero$ by returning $1$ if $d = d'$; otherwise it returns $0$.
      \item Lastly, $\Dist{}$ outputs its guess. \Adv{} outputs any bit $b$ that \Dist{} outputs.
    \end{enumerate}

    When \Adv{} is playing respectively \GamePRFZero{} and \GamePRFOne{}, then the reduction perfectly simulates $\Hyb_{-1}$ and $\HybZero$.
    Indeed, if \Adv{} is playing with \GamePRFZero{} (resp. \GamePRFOne{}) then, $\cS{0} = \PRFOracle(11 \concat \cl{0})$ (resp. \cS{0} is randomly sampled from $\bin^*$).
    In addition, \Adv{} computes all the secrets and edge information following~\Cref{const:dka}.
    As such, the advantage of the attacker \Adv{} in distinguishing \GamePRFZero{} and \GamePRFOne{} is non negligible.
    This concludes the proof.
\end{proof}
The rest of the proof is analogous to the one of Atallah \emph{et al.}, except that in $\Hyb^{(a)}_i$ we additionally sample $\cS{i}$ at random.
We refer to~\cite[Theorem 5.3]{atallah:2009:dek:1455526.1455531} for the proofs that $\Hyb_0 \approx_c \Hyb^{(a)}_{1}$ and $\Hyb^{(a)}_{i} \approx_c \Hyb^{(b)}_{i}$, $\Hyb^{(b)}_{i} \approx_c \Hyb^{(c)}_{i}$, $\Hyb^{(c)}_{i - 1} \approx_c \Hyb^{(a)}_{i}$ for any $i \in \set{2, \ldots, | \ancstar{} | - 1}$.

\end{proof}

\section{Handling Dynamic Changes to a Deterministic Key Assignment Access Hierarchy}\label{sec:dynamic_changes}

This section details how to handle dynamic changes to the access hierarchy (\emph{e.g.}, insertion of a node, or deletion of an edge) of our deterministic key assignment scheme of~\Cref{sec:deterministic_wallets_through_dynamic_and_efficient_key_management} and, in turn, within Arcula, our hierarchical deterministic wallet of~\Cref{sec:arcula}.

Handling dynamic changes to the access hierarchy of the DHKA requires us to consider two problems.
First, how to correctly enforce the hierarchy after the modification (\emph{e.g.}, preventing a node from accessing a subtree after an edge to that subtree is removed);
second, how to deal with modifications to the structure of \cG{} that change the path from the root to any node \cvi{} along which we deterministically derive the secret values \cSi{} (\emph{e.g.}, removing the parent of a node).
We solve these problems through the following strategies.
First, we modify the graph \cG{} by adding an explicit root node to it, \cvroot{}, such that there exists an edge between \cvroot{} and any \emph{root} node of \cG{} (\emph{i.e.}, any minimal node in a topological ordering of \cG{}).
More in details, we define $\cGfirst = (\cV{} \cup \set{\cvroot}, \cE{} \cup \set{(\cvroot, \cvi) \mid \cvi \text{ has no predecessors in \cG}})$.
It is easy to prove that both $\cG{}$ and \cGfirst{} define equivalent access hierarchies.

Next, we associate an additional identifier, that we call \emph{version}, to each node by including it in its label.
Let $\Ver{}: \cV \rightarrow \NN$ be a public mapping associating an integer $\cwi{} \in \NN$ to any node $\cvi{} \in \cV$.
Every node \cvi{} initially starts from version $w_i = 0$, and we modify~\Cref{item:difference_root,item:difference_node} of~\Cref{const:dka} to account for it when deriving the node label \cli{}:
\begin{align*}
  \cli &= \dlvi{}
\end{align*}
Every time we modify the graph \cGfirst{} in such a way that it would require updating the secret of a node \cvi{}, we do so by updating its version \cwi{}, deterministically computing its new label \cli{}, and, in turn, its new secret \cSi{}.

In the remainder of this section, we leverage the version associated to each node to perform a \emph{rekey} procedure, defined as follows for every node \cvh{} and for every node \cvp{} such that \cvp{} is the parent of \cvh{} in \cGfirst{}.
\rekey{}

Finally, we deal with the dynamic modifications of the graph:
\begin{description}
  \item[Deletion of an edge:] Let \sE{} be the edge that is to be removed from \cGfirst{}.
    Our goal is twofold:
    First, to prevent \cvi{} from accessing the cryptographic keys of \cvj{}.
    Second, to make sure that if the deletion of the edge changes the derivation path from the root \cvroot{} of the hierarchy to \cvj{}, then the deterministic generation of the secret \cSj{} changes accordingly.
    We begin by tackling this last problem.
    Let \cvp{} be the parent node of \cvj{} in \cGfirst{}.
    If $\cvi{} = \cvp{}$ and if there is no other edge $(\cvpfirst, \cvj{}) \in \cE$ (\emph{i.e.}, there does not exist another predecessor of \cvj{} that is a candidate to become its new parent), then the deletion of the edge \sE{} results in disconnecting of \cvj{} from the access hierarchy.
    In that case, we add a connecting edge $(\cvroot{}, \cvj{})$ to \cGfirst{} that creates a single-hop path from the root to \cvj{} and allows the deterministic key derivation of its secret \cSj{}.
    We note that the addition of this edge does not modify the access hierarchy, \emph{i.e.} it does not allow \cvj{} to derive the secrets of any node that was not previously between its descendants $\desc{(\cvj{})}$.

    Next, we prevent \cvi{} from accessing the cryptographic keys of \cvj{} by performing the \emph{rekey} procedure for each node $\cvh{} \in \desc{(\cvj)}$ (this includes \cvj{} as well).% and for each node \cvp{} such that \cvp{} is the parent of \cvj{} in the graph \cGfirst{} after removing the edge \sE{}.
  \item[Deletion of a node:] The deletion of any node \cvi{} corresponds to first removing all the incoming and outgoing edges of \cvi{} through the procedure specified above.
    Then, to removing the public and secret information associated with \cvi{} from the \Pub{}, \Sec{}, and \Ver{} data structures.
  \item[Insertion of an edge:] Let \sE{} be the edge to be included into \cGfirst{}.
    We consider two cases:
    \begin{itemize}
      \item After the addition of the edge $\cvi$ is the parent of $\cvj$ in \cGfirst{}.
        As before, we perform the \emph{rekey} procedure for each node $\cvh{} \in \desc{(\cvj)}$ to update their secret values and to allow the deterministic derivation.
        % and for each node \cvp{} such that \cvp{} is the parent of \cvj{} in the graph \cGfirst after inserting the edge \sE{}.
      \item Otherwise,
        compute \srij{}, \syij{}, and augment \Pub{} to contain the mapping $(\cvi, \cvj) \mapsto \cyij$.
    \end{itemize}
  \item[Insertion of a new node:] Let \cvi{} be the node to insert, together with a set of new edges in and out of it.
    Let \cvj{} be the parent of \cvi{}.
    We begin by computing a deterministic public label \sliver{} (where $\cwi{} = 0$) and a deterministic secret value \scSi{};
    then, we compute \ski{} and we augment \Pub{} with the mapping $\cvi \mapsto \cli$, \Sec{} with the mapping $\cvi \mapsto (\cSi, \cxi)$, and \Ver{} with the mapping $\cvi \mapsto 0$.
    Finally, we proceed to insert the edges one by one using the edge insertion procedure specified above.
  \item[Key Replacement:] To replace the cryptographic key \cxi{} associated to any node \cvi{}, we perform the \emph{rekey} procedure for each node $\cvh{} \in \desc{(\cvi)}$.% and for each parent \cvp{} of \cvh{}.
\end{description}

This approach allows our deterministic key assignment scheme to handle dynamic changes to its access hierarchy and requires the manager of the key assignment (\emph{e.g.}, a crypto-currencies exchange) to keep track of the version of the nodes stored within the \Ver{} mapping in addition to the structure of the graph \cG{}.
Because of the determinism of our scheme, every change to a node \cvj{} also propagates to all its descendants.
As an example, the replacement of its secret key \cxj{} requires incrementing its version $\cwj{}$ in the label \clj{} to compute a new secret value and a new cryptographic key.
In turn, this causes the secret information and the cryptographic keys of all its descendants \cvi{} to change as well (because of the deterministic derivation of the secret values \scSi{}).
The cost of such an update depends on the particular application and the structure of the access hierarchy.
If we use the DHKA to handle the keys associated with traditional Bitcoin transactions, for example, updating the cryptographic key of a node requires sending its funds to a new address and involves the payment of a transaction fee.
Most of the times, however, we are particularly interested in appending new leaves to the access hierarchy (\emph{e.g.}, to create a new node for an incoming payment).
This operation is a particular case of the insertion of a new node with a single incoming edge.
It never modifies any derivation path, and, as a consequence, it does not perform the \emph{rekey} procedure, it does not change any cryptographic key, and it does not require transactions on the blockchain.
Finally, when we use Arcula to enable the public derivation of addresses, we identify the nodes of the wallet through the master public key \mpk{} and a label.
Both the addition of a new node and the update of the label \clj{} of an existing node \cvj{} result in a new Arcula address (\emph{i.e.}, a locking script that contains a new label).
Spending the funds destined to the new address requires a new certificate, signed by the master secret key, that associates the new public key $\pkX_j'$ (obtained from the new intermediate key $x'_h$ after the \emph{rekey} procedure) to the new (or updated) label \clj{}.
\section{Time-Bound Deterministic Hierarchical Key Assignment}\label{sec:time_bound_deterministic_hierarchical_key_assignment}
A hierarchical key assignment scheme aims at assigning a cryptographic key to every user of an access hierarchy so that users with higher privileges can autonomously derive the keys of the others within their subtrees, \emph{i.e.}, with lower privileges in the hierarchy.
Many use cases require constraining these assignments according to some time restrictions.
For example, a service provider aims to provide a user with her cryptographic keys only as long as she pays for her subscription to the service.
To achieve this goal, it can leverage a key assignment scheme that takes time into account, and that enables the users to deriver their cryptographic keys during a given period only (\emph{e.g.}, one month).
This section details how we incorporate these temporal capabilities into the deterministic hierarchical key assignment scheme of~\Cref{sec:deterministic_wallets_through_dynamic_and_efficient_key_management} and within Arcula, our design hierarchical deterministic wallet (\Cref{sec:arcula}).

In the last few years, many researchers focused on how to incorporate temporal capabilities into HKA schemes~\cite{ateniese2012,atallah2007,desantis2008}.
The solutions proposed first modify the hierarchy of the assignment to consider, at the same time, both the access privileges and the temporal constraints.
Then, assign a set of secrets to the nodes of the augmented hierarchy so that the users can perform the key derivation according to the time constraints.

We add these constraints to our DHKA by relying on the work of De Santis \emph{et al.}~\cite{desantis2008} that shows how to design a time-bound key-indistinguishable HKA scheme from any provably secure HKA scheme (and, in particular, from our DHKA).
Let \sG{} be an access hierarchy and let $\cTT{} = \set{t_1, t_2, \ldots, t_n}$ be a sequence of distinct time periods.
Each user \cvi{} belongs to a node of the hierarchy for a non-empty contiguous subsequence $\cTTi{} = \set{t_j, \ldots, t_k} \subseteq \cTT$ of time periods.\footnote{In~\cite{desantis2008} the subsequence of time periods of a node $\cvi{} \in \cV$ is denoted by $\lambda_i$.}
Let $\cPP{} = \set{\cTTi{}}_{\cvi \in \cV}$ be the set of time subsequences \cTTi{} when every user $\cvi{} \in \cV$ belongs to the hierarchy.
The authors start from the observation that the contiguous subsequences $\cTTi{} \in \cPP$ implicitly define a partially ordered hierarchy, where $\cTTi{} < \cTTj{} \iff \forall t_k \in \cTTi \implies t_k \in \cTTj$, \emph{i.e.} iff \cTTi{} is included in \cTTj{}.
They call this relation the \emph{interval hierarchy}, and they use its minimal representation, where every node except the leaves has precisely two edges, to augment the original access hierarchy encoded by the graph \cG{}.
As a result, they build a new graph, $\cGt{} = (\cV_{\cTT{}}, \cE{}_{\cTT})$, that enforces both the access and the interval partially ordered hierarchies.
$\cGt{}$ contains a copy of the interval hierarchy for each node in \cG{}.
A user \cvi{} derives the cryptographic key assigned to its descendant \cvj{} for the period $t_k \in \cTTj{}$ by following the path in the augmented graph \cGt{} along the copy of the interval hierarchy related to \cvj{} and then through the original access hierarchy encoded by $\cG{}$.
The instantiation of the (D)HKA scheme on the graph $\cGt{}$ results in a (deterministic) time-bound hierarchical key assignment scheme.

By construction, the number of nodes and edges in \cGt{} grows quadratically in the size of \cTT{} and in the dimension of \cG{}.
In turn, the amount of public information required by a generic HKA scheme on \cGt{} grows comparably.
As we have seen in~\Cref{sub:dhka_implementation}, however, the determinism of our DHKA scheme allows us to reduce the amount of public information required significantly:
The nodes of the access hierarchy can derive the secret information of their descendants by leveraging their own secrets and only rely on the public information when a node has two or more predecessors.
In the same way, when we augment the access hierarchy encoded by \cG{} to account for the interval hierarchy into \cGt{}, the determinism of the scheme allows us to reduce the amount of public information required.
The augmented hierarchy \cGt{}, indeed, stores a copy of the minimal interval hierarchy for every node of \cG{}.
Every node of the minimal interval hierarchy only has a single predecessor and, as a result, does not require any public information associated with its edges.
For this reason, when we leverage our design of DHKA to incorporate the temporal capabilities into an HKA scheme, the size of the public information required grows only linearly with the dimension of the access hierarchy \cG{} and, in particular, is independent of the cardinality of \cTT{}.

To conclude, we show how to incorporate these temporal capabilities into Arcula, our design of HDW based on DHKA and digital signatures.
Our construction provides the users of the access hierarchy with a certificate and a signing key.
The certificate, signed by the master secret key, authorizes the signing key to spend the coins addressed to their identities.
When we add the temporal capabilities to the DHKA, we assign a different signing key to each user \cvi{} for each time period $t_j \in \cTTi{}$;
then, we provide her with a certificate $\cert_{i,j}$ for each key.
We prevent the users from signing new transactions through an outdated key by adding an expiration date to these certificates so that they are only valid until the end of the period $t_j$.
As a result, every user \cvi{} will require an updated certificate after each time period passes.
The stack-based scripting language of Bitcoin Cash does not allow yet to check for the expiration date of a certificate.
For this reason, our design of time-bound Arcula requires, at the time of writing, a more powerful scripting language, \emph{e.g.} an Ethereum smart contract.

% \section{Security proofs}\label{sec:missing_proofs}

% \subsection{Proof of \texorpdfstring{\Cref{th:d_key_ind}}{Theorem 4.1}}\label{subsec:proof_dhka}

\section{Proof of \texorpdfstring{\Cref{thm:arcula_euf}}{Theorem 5.1}}\label{subsec:proof_arcula}
We prove the theorem by contradiction, using a hybrid argument.
Let $(\cvj{}, m, \sigma)$ be the forgery returned by \Adv{} in the game \GameHSigForge{}.
We define the following hybrid experiments:
\begin{description}
    \item[\HybZero:] is exactly the game \GameHSigForge{}.
    \item[$\Hyb_t$:] is the same as $\Hyb_{t-1}$, except that the challenger generates at random the signature key pairs $(\skX_i, \pkX_i)$ for the first $t$ nodes in $\anc(\cvj{})$.
        More in details, let $\anc(\cvj) =\set{\cvzero, \ldots, \cvt{}, \ldots, \cvj{}}$, for every $\cvi{} \in \set{\cvzero, \ldots, \cvt{}}$ the challenger generates the signature key pair $(\skX_i, \pkX_i)$ by running $\Keygen_{\Sig}(\secparam)$.
\end{description}
The proof idea is to first show, using a hybrid argument, that $\HybZero{} \approx_c \Hyb_{|\anc(\cvj)|}$. Hence, a potential adversary $\Adv{}$ has the same advantage in both $\HybZero$ and $\Hyb_{|\anc(\cvj)|}$, with overwhelming probability.
Then, we show that an adversary $\Adv{}$ for $\Hyb_{|\anc(\cvj)|}$ implies an adversary $\Adv'{}$ for \GameSigForgeSig{}.

\begin{lemma}\label{lem:inductive_proof_cert_arcula}
    If \DHKA{} is key indistinguishable, then $\Hyb_{t-1} \approx_c \Hyb_{t}$ for every $1 \leq t \leq |\anc(\cvj)|$.
\end{lemma}
\begin{proof}
    We assume that there exists a DAG $\cG = (\cV, \cE)$ and a distinguisher \Dist{} that has a non-negligible advantage in distinguishing between $\Hyb_{t-1}$ and $\Hyb_{t}$.
    Then, we build an adversary \Adv{} against the experiment \GameSKeyIndDHKA{} (defined in~\Cref{def:sec_dka}) as follows:

    \begin{enumerate}
        \item \Adv{} samples at random \cvstar{}.
            Let $\ancstar{} = \set{\cvzero, \ldots, \cvt{}, \ldots, \cvstar}$ be the set of ancestors of \cvstar{} according to an ordering of the nodes of the graph (\emph{e.g.}, a topological sorting).
            \Adv{} sends \cvt{} to the challenger and receives \Pub{} and \cx{t}.
        % \item \Adv{} samples the pair $(\skK, \pkK) \sample \Keygen_{\Sig}(\secparam)$.
        \item \Adv{} executes the remaining steps of $\Set_{\Pi}$, except that it skips \Cref{enum:w_set_proof_skip} and it replaces~\Cref{enum:w_set_proof_inner} with the following:
            \begin{itemize}
                \item If $\cvi{} \in \set{\cvzero, \ldots, \cvtone{}}$, then compute $(\skX_i, \pkX_i) \sample \Keygen_{\Sig}(\allowbreak\secparam)$.
                \item Otherwise, if $\cvi{} = \cvt{}$, then compute $(\skX_t, \pkX_t) = \Keygen_{\Sig}(\allowbreak\secparam; \cx{t})$.
                \item Otherwise, send a $\CorrOracle^\DHKA(\cvi)$ query to the challenger and receive $\cSi{} = \derkeyi{}$.
                Compute $\cxi{} = \swDeriveDHKAi{}$ and $(\skX_i, \pkX_i) = \Keygen_{\Sig}(\secparam; \cxi)$.
            \end{itemize}
            Finally, \Adv{} outputs the public parameters $\pp{} = (\cG{}, \Pub{}, \crtset{}, \pkX_0)$.
        \item \Adv{} answers oracle queries in the following way:
        \begin{itemize}
            \item On input $\cvi{}$ for $\CorrOracle^{\Pi}$, \Adv{} invokes $\CorrOracle^{\DHKA}(\cvi{})$ and returns the output.
            \item On input $(m, \cvi{})$ for $\SignOracle^{\Pi}$, \Adv{} returns $\sigma = (\pkX_i, \sigma', \hatsigmai)$ where $\sigma' \sample \Sign_{\Sig}(\skX_i, m)$.
        \end{itemize}
        \item \Adv{} receives the forgery $(\cvj{}, m, \sigma)$.
            It aborts the simulation if $\cvstar \neq \cvj{}$; otherwise it completes the simulation by returning the result of $\Verify_{\Pi}(\pk_j, m, \sigma)$, where $\clj = \Pub(\cvj)$ and $\pk_j = (\pkX_0, \clj{})$.
        \item \Adv{} outputs the decisional bit received from \Dist{}.
    \end{enumerate}
    Let \eabort{} be the event that \Adv{} aborts the simulation.
    It is easy to see that $\prob{\enabort{}} = \prob{\cvstar{} = \cvj{}} = \frac{1}{\mid \cV{} \mid}$.
    Let \GameSKeyIndB{} be the key indistinguishability game with bit $b$.
    Conditioned on the event \enabort{}, when \Adv{} is playing respectively \GameSKeyIndZero{} and \GameSKeyIndOne{}, then the reduction perfectly simulates $\Hyb_{t-1}$ and $\Hyb_{t}$, because \Dist{} can not corrupt any node $\cv{} \in \ancstar$.
    Hence, the advantage of the attacker \Adv{} in winning the game \GameSKeyIndDHKA{} is non-negligible.
    This concludes the proof.
\end{proof}

\begin{lemma}\label{lem:reduction_proof_cert_arcula}
    If \Sig{} is existentially unforgeable, then for every DAG $\cG = (\cV, \cE)$ and PPT adversary $\Adv{}$, $\prob{\Hyb_{|\anc(\cvj)|, \Adv}(\secpar, \cG) = 1} \leq \negl$.
\end{lemma}
\begin{proof}
We assume that there exists a DAG $\cG = (\cV, \cE)$ and an adversary \Adv{} that has a non-negligible advantage against $\Hyb_{|\anc(\cvj)|, \Adv}(\allowbreak\secparam, \cG)$.
Then, we build an adversary $\Adv'{}$ against \GameSigForgeSig{} as follows:
\begin{enumerate}
    \item $\Adv'{}$ receives $\pk^*$ from the challenger.
    \item $\Adv'{}$ flips a bit $d \sample \bin$ and samples at random $\cvstar{} \sample \cV$ and $\cSeed \sample \cSS{}$.
    \item $\Adv'{}$ simulates $\Set_{\Pi}$. It runs $\PubSec{} = \sdSetDHKA{}$.
        If $d = 0$, it sets $\pkX_0 = \pk^*$; otherwise it runs $(\skX_0, \pkX_0) = \KeygenECDSA(\secparam;\cx{0})$ where $(\cS{0}, \cx{0}) = \Sec(\cv{0})$.
        Lastly, $\Adv'{}$ executes the remaining steps of $\Set_{\Pi}$, except that it replaces~\Cref{enum:w_set_proof_inner} and~\Cref{enum:sign_oracle_proof_inner} with the following:
        \begin{description}
            \item[\Cref{enum:w_set_proof_inner}:] $\Adv'{}$ proceeds as follow:
                \begin{itemize}
                    \item If $\cvi \in \ancstar{} \setminus \set{\cvstar{}}$, then compute $(\skX_i, \pkX_i) \sample \allowbreak \Keygen_{\Sig}(\secparam)$.
                    \item If $\cvi =\cvstar{}$, set $\pkX_i = \pk^*$ if $d=1$; otherwise run $(\skX_i, \pkX_i) \sample \Keygen_{\Sig}(\secparam)$.
                    \item Otherwise (if $\cvi \not\in \ancstar{}$), run $(\skX_i, \pkX_i) = \Keygen_{\Sig}(\allowbreak\secparam; \cxi)$ where $(\cSi, \cxi) = \Sec(\cvi)$.
                \end{itemize}
            \item[\Cref{enum:sign_oracle_proof_inner}:] If $d=1$, then retrieve the label $\cli{} = \Pub(\cvi)$ and compute $\hatsigmai \sample \Sign_{\Sig}(\skX_0, (\pkX_i, \cli))$; otherwise, set $\hatsigmai \sample \SignOracle^\Sigma((\pkX_i, \cli))$.
        \end{description}
        Finally, $\Adv'{}$ sends to $\Adv{}$ the public parameters $\pp{} = (\cG{}, \Pub{}, \crtset{}, \pkX_0)$.
    \item $\Adv{}'$ answers oracle queries in the following way:
        \begin{itemize}
            \item On input $\cvi{}$ for $\CorrOracle^{\Pi}$, $\Adv{}'$ returns $\derkeyi=\cSi$ where $(\cSi, \cxi) = \Sec(\cvi)$.
            \item On input $(m, \cvi{})$ for $\SignOracle^{\Pi}$, if $d=1 \land \cvi = \cvstar{}$, $\Adv'{}$ sets $\sigma' \sample \SignOracle^\Sigma(m)$; otherwise, it computes $\sigma' \sample \Sign_{\Sig}(\skX_i,\allowbreak m)$.
                Lastly, it returns $\sigma = (\pkX_i, \sigma', \hatsigmai)$.
        \end{itemize}
    \item \Advfirst{} receives the forgery $(\cvj{}, \mforgery, \sigmaforgerytilde)$ such that $\sigmaforgerytilde = (\pkforgery_j, \sigmaforgery, \hatsigmaforgery_j)$ and aborts the simulation if $\cvstar \neq \cvj{} \lor (d = 0 \land \pkforgery_j = \pkX_j{}) \lor (d = 1 \land \pkforgery_j \neq \pkX_j{})$.
    Otherwise, if $d=0$, it sends the forgery $((\pkforgery_j, \clj),\hatsigmaforgery_j)$ to challenger where $\clj{} = \Pub(\cvj)$; if $d=1$ sends $(\mforgery,\sigmaforgery)$.
\end{enumerate}
Let \eabort{} be the event that \Advfirst{} wins the game \GameSigForgeSig{} and aborts the simulation.
First of all, note that:
\ifFULL
\begin{align*}
    \enabort{} &= \lnot\left[\cvstar \neq \cvj{} \lor (d = 0 \land \pkforgery_j = \pkX_j{}) \lor (d = 1 \land \pkforgery_j \neq \pkX_j{}) \right]\\
    &= \left[\cvstar = \cvj{} \land \lnot(d = 0 \land \pkforgery_j = \pkX_j{}) \land \lnot (d = 1 \land \pkforgery_j \neq \pkX_j{}) \right] \\
    &= \left[\cvstar = \cvj{} \land (d = 1 \lor \pkforgery_j \neq \pkX_j{}) \land (d = 0 \lor \pkforgery_j = \pkX_j{}) \right] \\
    &= \left[\cvstar = \cvj{} \land ((d = 0 \land d = 1) \lor (d=1 \land \pkforgery_j = \pkX_j{}) \right.\\
    &\left. \qquad \lor (d = 0 \land \pkforgery_j \neq \pkX_j{}) \lor (\pkforgery_j \neq \pkX_j{} \land \pkforgery_j = \pkX_j{}))\right] \\
    &=\left[\cvstar = \cvj{} \land ((d=1 \land \pkforgery_j = \pkX_j{}) \lor (d = 0 \land \pkforgery_j \neq \pkX_j{}))\right]
\end{align*}
\else
\begin{align*}
    \enabort{} &= \lnot\left[\cvstar \neq \cvj{} \lor (d = 0 \land \pkforgery_j = \pkX_j{}) \lor \right.\\
    & \qquad \left. (d = 1 \land \pkforgery_j \neq \pkX_j{}) \right]\\
    &= \left[\cvstar = \cvj{} \land (d = 1 \lor \pkforgery_j \neq \pkX_j{}) \land \right.\\
    & \qquad \left. (d = 0 \lor \pkforgery_j = \pkX_j{}) \right] \\
    &= \left[\cvstar = \cvj{} \land ((d = 0 \land d = 1) \lor \right.\\
    & \qquad (d=1 \land \pkforgery_j = \pkX_j{}) \lor (d = 0 \land \pkforgery_j \neq \pkX_j{}) \\
    & \qquad \lor \left. (\pkforgery_j \neq \pkX_j{} \land \pkforgery_j = \pkX_j{}))\right] \\
    &=\left[\cvstar = \cvj{} \land ((d=1 \land \pkforgery_j = \pkX_j{}) \lor \right.\\
    & \qquad \left. (d = 0 \land \pkforgery_j \neq \pkX_j{}))\right]
\end{align*}
\fi
Let $\prob{\pkforgery_j = \pkX_j} = p$. We can express $\prob{\enabort{}}$ in the following way:
\ifFULL
\begin{align*}
    \prob{\enabort{}} &= \prob{\cvstar = \cvj{} \land (d = 0 \land \pkforgery_j \neq \pkX_j{}) \lor (d = 1 \land \pkforgery_j = \pkX_j{})} \\
    &= \prob{\cvstar = \cvj{}} \cdot \left(\prob{d = 0 \land \pkforgery_j \neq \pkX_j{}} + \prob{d = 1 \land \pkforgery_j = \pkX_j{}}\right) \\
    &= \prob{\cvstar = \cvj{}} \cdot \left(\prob{d = 0} \cdot \prob{\pkforgery_j \neq \pkX_j{}} + \prob{d = 1} \cdot \prob{\pkforgery_j = \pkX_j{}}\right) \\
    &= \frac{1}{\mid \cV{} \mid} \cdot \left(\frac{1 - p}{2}  + \frac{p}{2}\right) = \frac{1}{2 \cdot \mid \cV{} \mid}
\end{align*}
\else
\begin{align*}
    \prob{\enabort{}} &= \Pr\left[\cvstar = \cvj{} \land (d = 0 \land \pkforgery_j \allowbreak\neq \pkX_j{}) \lor \right.\\
    & \left. \qquad(d = 1 \land \pkforgery_j = \pkX_j{}) \right]\\
    &= \prob{\cvstar = \cvj{}} \cdot \left(\prob{d = 0 \land \pkforgery_j \neq \pkX_j{}} + \right.\\
    \allowbreak
    &\left. \qquad \prob{d = 1 \land \pkforgery_j = \pkX_j{}}\right) \\
    &= \prob{\cvstar = \cvj{}} \cdot \left(\prob{d = 0} \cdot \prob{\pkforgery_j \neq \pkX_j{}} \right.\\
    &\left. \qquad + \prob{d = 1} \cdot \prob{\pkforgery_j = \pkX_j{}}\right) \\
    &= \frac{1}{\mid \cV{} \mid} \cdot \left(\frac{1 - p}{2}  + \frac{p}{2}\right) = \frac{1}{2 \cdot \mid \cV{} \mid}
\end{align*}
\fi
Let $\cQQ^{\Sig}_{\Sign}$ and $\cQQ^{\Pi}_{\Sign}$ be respectively the set of queries submitted by $\Adv'{}$ to $\SignOracle^{\Sig}$ and the set of queries submitted by $\Adv{}$ to $\SignOracle^{\Pi}$.
Conditioned on \enabort{} and since $\Adv{}$ is a valid adversary for \GameHSigForge, then, with non-negligible probability, $\Verify_{\Pi}(\pk_j, \mforgery, \sigmaforgerytilde) = 1$ if and only if $\Verify_{\Sig}(\pkK, (\pkforgery_j, \clj{}), \hatsigmaforgery_j) =1$ and $\Verify_{\Sig}(\pkforgery_j, \allowbreak \mforgery, \sigmaforgery) = 1$, where $\pk_j = (\pkX_0, \clj)$ and $\clj = \Pub(\cvj)$.
Note that $\Adv'{}$ outputs a valid forgery for $\GameSigForgeSig{}$ with probability $\frac{1}{2}$:
\begin{enumerate}
\item Whenever $d=0$, we have $\pkX_0 = \pk^*$ and $\pkforgery_j \neq \pkX_j$.
This allows us to conclude that $\Adv'{}$ never asked $(\pkforgery_j, \clj)$ to oracle $\SignOracle^{\Sig}$ (\emph{I.e.,} $(\pkforgery_j, \clj) \not\in \cQQ^{\Sig}_{\Sign}$). Hence, $((\pkforgery_j, \clj),\hatsigmaforgery_j)$ is a valid forgery for $\GameSigForgeSig{}$.
\item On the other hand, if $d=1$, we have $\pkforgery_j = \pkX_j = \pk^*$. Since, $\Adv{}$ is a valid adversary it must produces a valid signature for a new fresh message. Hence, we can conclude that $(\cvstar{}, \mforgery) \not\in \cQQ^{\Pi}_{\Sign}$ and $(\mforgery, \sigmaforgery)$ is a valid forgery for $\GameSigForgeSig{}$.
\end{enumerate}

\noindent This concludes the proof.
\end{proof}

By combining \Cref{lem:inductive_proof_cert_arcula} and \Cref{lem:reduction_proof_cert_arcula} we have that \Cref{const:cert-arcula} is hierarchically existentially unforgeable.

\end{document}